\documentclass[letterpaper, 10pt, conference]{IEEEtran} 
\makeatletter
\def\blfootnote{\xdef\@thefnmark{}\@footnotetext}
\makeatother
\normalsize

\DeclareMathAlphabet\mathbfcal{OMS}{cmsy}{b}{n}

\usepackage{amsthm}
\usepackage{bbm}
\usepackage{algorithm}
\usepackage{algorithmic}
\usepackage{graphicx,amsmath,amssymb,latexsym,epsfig,url}

\DeclareMathOperator*{\esssup}{ess\,sup}
\usepackage{setspace}
\usepackage{amscd}
\usepackage{psfrag}
\usepackage{array,booktabs,arydshln,xcolor,cite}

\usepackage{lscape}
\newtheorem{theorem}{Theorem}
\newtheorem{proposition}{Proposition}

\newtheorem{corollary}{Corollary}
\newtheorem*{proof of Theorem*}{Proof of Theorem 3}
\newtheorem{proof of Lemma}{Proof of Lemma}
\newtheorem{definition}{Definition}
\newtheorem{lemma}{Lemma}

\usepackage{color}

\begin{document}
%\addtolength{\textheight}{1cm}
\vspace{-0.2 in}
\title{Optimal Status Updating with a Finite-Battery Energy Harvesting Source}
\author{\IEEEauthorblockN{Baran Tan Bacinoglu\IEEEauthorrefmark{1}, Yin Sun\IEEEauthorrefmark{3}, Elif Uysal\IEEEauthorrefmark{1}, and Volkan Mutlu\IEEEauthorrefmark{1}}
\IEEEauthorblockA{\IEEEauthorrefmark{1}METU, Ankara, Turkey,
\IEEEauthorrefmark{3}Auburn University, AL, USA\\
 E-mail:  barantan@metu.edu.tr, yzs0078@auburn.edu, uelif@metu.edu.tr,  volkan.mutlu@metu.edu.tr}
}

\bibliographystyle{IEEEtran}

\maketitle
\vspace{-0.1 in}
\def\eg{\emph{e.g.}}
\def\ie{\emph{i.e.}}

\begin{abstract}
We consider an energy harvesting source equipped with a finite battery, which needs to send timely status updates to a remote destination. The timeliness of status updates is measured by a non-decreasing penalty function of the Age of Information (AoI). The problem is to find a policy for generating updates that achieves the lowest possible time-average expected age penalty among all online policies.  We prove that one optimal solution of this problem is a monotone threshold policy, which satisfies (i) each new update is sent out only when the age is higher than a threshold and (ii) the threshold is a non-increasing function of the instantaneous battery level. Let $\tau_B$ denote the optimal threshold corresponding to the full battery level $B$, and $p(\cdot)$ denote the age-penalty function, then we can show that $p(\tau_B)$ is equal to the optimum objective value, i.e., the minimum achievable time-average expected age penalty. These structural properties are used to develop an algorithm to compute the optimal thresholds. Our numerical analysis indicates that the improvement in average age with added battery capacity is largest at small battery sizes; specifically, more than half the total possible reduction in age is attained when battery storage increases from one transmission's worth of energy to two. This encourages further study of status update policies for sensors with small battery storage.

\end{abstract}
\begin{IEEEkeywords}Age of information; age-energy tradeoff; non-linear age penalty, threshold policy; optimal threshold; energy harvesting; battery capacity.\end{IEEEkeywords}
\section{Introduction}
\blfootnote{This paper was presented in part at IEEE ISIT 2018 ~\cite{8437573}.
This work was supported in part by NSF grant CCF-1813050, ONR
grant N00014-17-1-2417 and TUBITAK grant no 117E215.}
The \emph{Age of Information} (AoI), or simply the age, was proposed in \cite{ Kaul2011, Kaul2012} as a performance metric that measures the freshness of information in status-update systems. For a flow of information updates sent from a source  to a destination, the age is defined as the time elapsed since the newest update available was generated at the source. That is, if $U(t)$ is the largest among the time-stamps of all packets received  by time $t$, the age is defined as:
\begin{equation}
\Delta(t)= t-U(t),
\end{equation}
 AoI is a particularly relevant performance metric for status-update applications that have growing importance in  remote monitoring \cite{ZviedrisESMS10, blueforce}, machine-type communication, industrial manufacturing, telerobotics, Internet of Things and social networks. 

In many applications, the timeliness of status updates  not only  determines the quality of service, but also  affects other design goals such as the controllability of a dynamical system that relies on the updates  of sensing and control signals. AoI quantifies the timeliness of status-updates from the perspective of the receiver rather than throughput or delay based measures that are actually channel-centric.  Moreover, AoI is also related to measures such as the time-average mean-square error (MSE) for remote estimation. An example of this is the result in \cite{YinSunISIT2017} which showed remote estimation of a Wiener process minimizing MSE reduces to an AoI optimization problem when the sampling times at the transmitting side are independent of the  process. While AoI optimization based on linear functions of the age $\Delta(t)$ is a relevant performance goal for most scenarios, the performance of some applications may be related to \emph{non-linear} functions of the age. For example, the change in the value of stale data can be less/more significant as its age grows. In such cases, the penalty of data staleness can be modelled as a non-linear function $p(\Delta(t))$ of the age $\Delta(t)$, i.e., the \emph{age-penalty}. This function is chosen to be non-decreasing so that a decrease in \emph{age-penalty} can be only possible when the age is less. Accordingly, the optimization of the age-penalty parallels to average AoI optimization while it might have distinct optimality conditions.

 Ideally, AoI is minimized when status updates are frequent and fresh. That is, good AoI performance requires packets with low delay received regularly. A limitation in the minimization of AoI is a constraint on the long-term average update rate which may be due to an average power budget for the channel over which status updates are sent.  A stricter constraint is to keep a detailed budget on the number of status updates by allowing  update transmission when  a replenishable resource becomes available. This is the case of energy harvesting communication systems where each update consumes a certain amount of the harvested energy, if available. In the related literature of AoI optimization for energy harvesting communication systems, energy harvesting process is considered as an arrival process where each energy arrival carries the energy required for an update \cite{TanITA2015, TanISIT2017, WuYang2017, 2015ISITYates, 2018arXiv180202129A, DBLP:journals/corr/abs-1806-07271, 8437547, 8636088}. The goal of AoI optimization in such formulations is to find an optimal timing of update instants in order to minimize average AoI while transmission opportunities are subject to the availability of energy. Energy arrivals occur irregularly or randomly, which models an energy harvesting scenario.
 The main challenge in optimizing time average expected age under random energy arrivals is that in the case of an energy outages (empty battery), the transmitter must idle for an unknown duration of time. If it is the case that such random durations are inevitable, they introduce a tension for the regulation of inter-update durations. Another challenge is due to the finiteness of battery sizes. Theoretically, it is possible to achieve asymptotically optimal average AoI by employing simple schemes assuming infinite \cite{TanISIT2017} or sufficiently large battery \cite{WuYang2017} sizes. However, when the battery size is comparable to the energy required per update, such simple schemes do not allow performance guarantees. Consequently, it is important to explore optimal policies under such regimes where  performance depends heavily on the statistics of energy arrivals and the battery size.    

This study is motivated by the aforementioned challenges of optimizing AoI in energy harvesting systems, capturing both the randomness of energy arrivals and finite energy storage capability. In addition capturing both challenges we go further, by optimizing not only average age itself, but a more general age penalty function $p(\Delta(t))$ that is not necessarily linear (see \cite{Cho:2003:EPR:958942.958945, YinSunInfocom2016, 8000687, 8006543, DBLP:journals/corr/abs-1812-07241, Razniewski:2016:OUF:2983323.2983719}). Hence, the problem considered in this study is an age-penalty optimization problem where status updates consume discrete units of energy that are randomly generated, \ie, harvested, such that the number of energy units that can be stored at a time is limited by a finite value which is called battery capacity.

Under the assumption of Poisson energy arrivals, we show the structure of solutions for the age-penalty optimization problem. The structure of the optimal solution reflects a basic intuition about the optimal strategy: Updates should be sent when the update is valuable (when the age is high) and the energy is cheap (the battery level is high). We show that the optimal solution is given by a stopping rule according to which an update is sent when its immediate cost is surpassed by the expected future cost. For Poisson energy arrivals, this stopping rule can be found in the set of policies that we refer as \emph{monotone threshold policies}. Monotone threshold policies have the property that each update is sent only when the age is higher than a certain threshold which is a non-increasing function of the instantaneous battery level. One of our key results is that the value of the age-penalty function at the optimal threshold corresponding to the full battery level is exactly equal to the optimal value of the average age-penalty.

% The results for the structure of these solutions extend the results obtained for average AoI objective to non-linear age-penalty objective while considering arbitrary battery sizes.   

%The linear version of AoI optimization problem we considered in this study was formulated in \cite{TanISIT2017}, and concurrently in \cite{WuYang2017} where the special cases of unit battery and infinite battery capacity were investigated.  This paper extends \cite{TanISIT2017} and provides results for arbitrary battery sizes. 

\subsection{Contributions}
%This paper extends~\cite{TanISIT2017}, making the following contributions:
The contributions of this paper can be summarized as follows:
\begin{itemize}
\item We formulate the general average age-penalty optimization problem for sending status updates from an energy harvesting source. This generalizes the AoI optimization goal in the prior studies \cite{TanISIT2017, WuYang2017, 2018arXiv180202129A, DBLP:journals/corr/abs-1806-07271, 8437573} to a non-linear function of age. In addition to the generalization on the objective, the optimization is carried out over a more general policy space defined only using the causality assumption. We prove that solutions to this general optimization problem can be found among threshold-type policies.  
\item We show that, for optimal threshold-type policies with non-decreasing thresholds, the value of the penalty function at the threshold corresponding to the highest battery level is equal to the minimum value of the average age-penalty. As this optimal threshold is also the minimum of optimal thresholds at different battery levels, this implies that inter-update durations under such a policy is always above the minimum value of the average age-penalty.
\item For the case when the age-penalty function is linear, i.e., average AoI minimization problem, we provide the optimal thresholds  for integer battery size up to $4$.  These results show that  the most significant decrease in the minimum average AoI happens when incrementing the battery capacity of unit size (capable of holding one packet transmission's worth of energy) to two units.  The minimum achievable average AoI with a battery size of 4 units is only about $10\%$ larger than the ultimate minimum average AoI with infinite battery capacity. That is a promising result for small sensor systems.
\item For average AoI minimization problem, we provide an algorithm that can find near optimal policies achieving average AoI values arbitrarily close to the optimal values for any given battery capacity. This algorithm provides a methodical way to derive near optimal policies utilizing analytical results.     
\end{itemize}
\subsection{Paper Organization}
The rest of the paper is organized as follows. In Section \ref{sec:relatedwork}, the related work is discussed and summarized.  In Section \ref{sec:sm}, the system model and the formulation of the AoI optimization problem are described. In Section \ref{sec:mainresults}, the main results on the structural properties of the solution to the AoI optimization problem are shown and an algorithm to derive solutions for arbitrary integer battery sizes is provided. In Section \ref{sec:nm}, the numerical results validating analytical results and also showing optimal solutions for integer battery size up to $4$ are presented. In Section \ref{sec:conc}, the paper is concluded summaring the results and insights obtained over the course of this study.
\section{Related Work}
\label{sec:relatedwork}
Several studies on AoI considered this performance metric under various queueing system models comparing service disciplines and queue management policies (\eg, \cite{Ephremides2013, Ephremides2014, Huang2015, Pappas2015, Ephremides2016, Najm2016, 8469047, StatusUpdateHARQ, 8406846}). A common observation in these studies was that many queueing/service policies that are throughput and delay optimal but are often suboptimal with respect to AoI, while AoI-optimal policies can be throughput and delay optimal, at the same time. This showed that AoI optimization is  quite different than optimization with respect to classical performance metrics. This required many queueing models to be re-addressed under respect to age related objectives. Moreover, queueing system formulations typically assume no precise control on the transmission or generation times of status updates. However, such control is important for age optimization \cite{YinSunInfocom2016, 8000687}.

A direct control on the generation times of status updates is possible through a control algorithm that runs at the source. This  is the \emph{``generate-at-will"} assumption formulated in \cite{TanITA2015, 2015ISITYates} and studied in \cite{YinSunInfocom2016, YinSunISIT2017, 8000687}.  In \cite{TanITA2015}, the problem of AoI optimization for a source, which is constrained by an arbitrary sequence of energy arrivals was studied. In \cite{2015ISITYates}, AoI optimization was considered for a source that harvests energy at a constant rate under stochastic  delays experienced by the status update packets. The results in these studies showed suboptimality of work-conserving transmission schemes. Often, introducing a waiting time before sending the next update is optimal. That is, for maximum freshness, one may sometimes send updates at a rate lower than one is allowed to which may be counter-intuitive at first sight.

%%%%%%%%

The problem in \cite{TanITA2015} was extended to a continuous-time formulation with Poisson energy arrivals, finite energy storage (battery) capacity, and random packet errors in the channel in \cite{TanISIT2017}. An age-optimal threshold policy was proposed for the unit battery case, and the achievable AoI for arbitrary battery size was bounded  for a channel with a constant packet erasure probability. The concurrent study in \cite{WuYang2017}, limited to the special cases of unit battery capacity and infinite battery capacity computed the same threshold-type policies under these assumptions.
These special cases were investigated also for noisy channels with a constant packet erasure probability in \cite{8437547, 8636088}. The case for a battery with 2-units capacity was studied in \cite{2018arXiv180202129A} and the optimal policies for this case characterized as threshold-type policies similar to the optimal policy for unit battery capacity introduced in \cite{TanISIT2017} and \cite{WuYang2017}. Optimal policies for arbitrary battery sizes were characterized via Lagrangian approach in  \cite{DBLP:journals/corr/abs-1806-07271} and using optimal stopping theory in \cite{8437573}.  

%In \cite{twohop_energyharvesting}, the offline results in \cite{TanISIT2017} were extended considering fixed non-zero service time and the result is used to obtain a solution for the two-hop scenario. Another offline problem under energy harvesting was investigated in \cite{ArafaDelay} where the transmission delay of an update is controlled by energy consumed on its transmission.  

%%%%%%%%

%\vspace{-0.1 in}
\section{System Model}
%\vspace{-0.1 in}
\label{sec:sm}
Consider an energy harvesting transmitter that sends update packets to a receiver, as illustrated in Fig \ref{sensormodelw}. Suppose that the transmitter has a finite battery which is capable of storing up to $B$ units of energy. Similar to \cite{TanISIT2017}, we assume that the transmission of an update packet consumes one unit of energy. 
The energy that can be harvested arrive in units according to a Poisson process with rate $\mu_H$. Let $E(t)$ denote the amount of energy stored in the battery at time $t$ such that $0\leq E(t) \leq B$. 
The timing of status updates are controlled by a sampler which can monitor the battery level $E(t)$ for all $t$. We assume that the initial age and the initial battery level are zero, i.e., $\Delta(0)=0$ and $E(0)=0$.
\begin{figure}[htpb]
    \centering \includegraphics[scale=0.72]{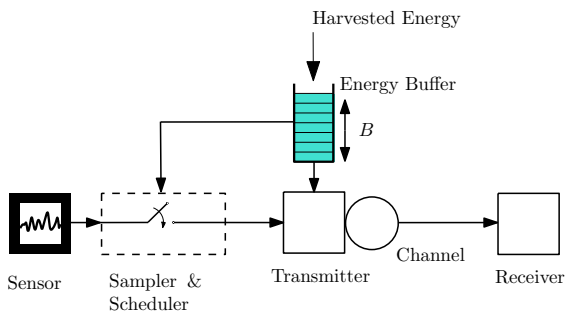}
\caption{System Model.}
%\vspace{-0.099in}
\label{sensormodelw} 
\end{figure}
 
Let $H(t)$ and $A(t)$ denote the number of energy units that have arrived during $[0,t]$ and the number of updates sent out during $[0,t]$, respectively. Hence, $\{H(t), t\geq 0\}$ and $\{A(t),t\geq0\}$ are two counting processes. If an energy unit arrives when the battery is full, it is lost because there is no capacity to store it.

The system starts to operate at time $t=0$. Let $Z_k$ denote the generation time of the $k$-th update packet such that $0=Z_0 \leq Z_1\leq Z_2\leq\ldots$. An update policy is represented by a sequence of update instants $\pi=(Z_{0},Z_{1},Z_{2},...)$. Let $X_{k}$ represent the inter-update duration between updates $k-1$ and $k$, i.e., $X_{k}=Z_{k}-Z_{k-1}$.  In many status-update systems  (\eg, a sensor reporting temperature \cite{5483217}), update packets are small in size and are only sent out sporadically. Typically, the duration for transmitting a packet is much smaller than the difference between two subsequent update times, i.e., $X_{k}$s are typically large compared to the duration of a packet transmission. With such systems in mind, in our model, we will approximate the packet transmission durations as zero. In other words, once the $k$-th update is generated and sent out at time $t=Z_k$, it is immediately delivered to the receiver. Hence, the age of information $\Delta(t)$ at any time $t\geq0$ is 
\begin{equation}
\label{efuage}
\Delta (t) = t-\max\{Z_k: Z_k \leq t\},
\end{equation}
which satisfies $\Delta(t) = 0$ at each update time $t=Z_k$. Because an update costs one unit of energy, the battery level reduces by one upon each update, i.e.,
\begin{equation}
\label{energydecrease}
E(Z_{k})=E(Z_{k}^{-})-1,
\end{equation}
where $Z_{k}^{-}$ is the time immediately before the $k$-th update. Further, because the battery size is $B$, the battery level evolves according to 
\begin{align}
\label{energyavailability}
E(t)=\min \lbrace E(Z_{k})+H(t)-H(Z_{k}),B\rbrace, 
\end{align}
when  $t \in [Z_{k}, Z_{k+1})$ is between two subsequent updates.

In terms of energy available to the scheduler, we can define update policies, that do not violate causality, as in the following:

\begin{definition}
\normalfont
A policy $\pi$ is said to be \emph{energy-causal} if updates only occur when the battery is non-empty, that is, $E(Z_{k}^-)\geq 1$ for each packet $k$.
\end{definition}

Another restriction on update instants is due to the information available to the scheduler which we define as follows,
\begin{definition}
\normalfont
Information on the energy arrivals and updates by time $t$ is represented by the filtration \footnote{Note that the filtration is right continuous as both $H(t)$ and $A(t)$ are right continuous.} $\mathcal{F}_{t}=\sigma(\lbrace( H(t'),A(t')),0 \leq t' < t  \rbrace)$ which is the $\sigma$-field generated by the sequence of energy arrivals and updates, i.e., $\lbrace( H(t'),A(t')),0 \leq t' < t  \rbrace$. 
\end{definition}
Similar to the definition of energy-causal policies, in the policy space that we will consider we merely assume the causality of available information besides energy causality. To formulate this assumption, we use the definition of $\mathcal{F}_{t}$. In terms of information available to the scheduler, any random time instant $\theta$ does not violate causality  if and only if $\left\lbrace \theta \leq t\right\rbrace \in \mathcal{F}_{t}$ for all $t\geq 0$. We will refer such random instants as \emph{Markov times}\cite{peskir2006optimal} and consider update times as Markov times based on the filtration $\mathcal{F}_{t}$ in general. Notice that such update times do not have to be finite, however, we will refer Markov times that are also finite with probability 1 (w.p.1.) as \emph{stopping times}\cite{peskir2006optimal}. For a policy trying to regulate age, it is legitimate to assume that update instants are always finite w.p.1. as otherwise the age may grow unbounded with a positive probability. With this in mind, we will consider only the update instants that are stopping times. 

Accordingly, we can define the \emph{online} update policies combining the causality assumptions on available energy and information as follows:
\begin{definition}
\normalfont
A policy is said to be \emph{online} if (i) it is energy causal, (ii) no update instant is determined based on future information, i.e., all update times are stopping (finite Markov) times based on  $\mathcal{F}_{t}$, i.e., $Z_{k}$ is finite w.p.1. while  $\left\lbrace Z_{k} \leq t\right\rbrace \in \mathcal{F}_{t}$ for all $t\geq 0$ and $k \geq 1$.
\end{definition}
Let $\Pi^{\mathsf{online}}$ denote  the set of online update policies.
To evaluate the performance of online policies, we consider an \emph{age-penalty function} that relates the age at a  particular time to a cost which increases by the age. This  function is defined as in below:

We consider an \emph{age-penalty function} $p(\cdot)$ that maps the age $\Delta(t)$ at time $t$ to a penalty $p(\Delta(t))$:
\begin{definition}

\normalfont
A function $ p: [0,\infty) \rightarrow [0,\infty)$ of the age is said to be an \emph{age-penalty function}  if 
\begin{itemize}
\item $\lim_{\Delta\rightarrow \infty}p(\Delta)=\infty$.
\item $p(\cdot)$ is a non-decreasing function.
\item $\int_0^\infty p(t)e^{-\alpha t} dt<\infty$ for all $\alpha> 0$.
\end{itemize}

\end{definition}
Observe that the definition of age-penalty functions covers any non-decreasing function of age that is of sub-exponential order\footnote{This is due to the third property in the definition, which is a technical requirement for the proofs.}  and grows to infinity. 

The time-average expected value of the age-penalty or simply the \emph{average age-penalty} can be expressed as 

\begin{equation}
\label{pavaged}
\bar{p}=\limsup_{T \rightarrow \infty}\frac{1}{T}\mathbb {E}\left[ \displaystyle\int_{0}^{T}p(\Delta(t))  dt \right].
\end{equation}

Let $\bar{p}_{\pi}$ denote the average age-penalty achieved by a particular policy $\pi$. The goal of this paper is to find the optimal update policy for minimizing the average age-penalty, which is formulated as
\vspace{-0.05 in}
\begin{equation}
\label{minavaged}
\displaystyle\min_{\pi\in\Pi^{\mathsf{online}}}\bar{p}_{\pi}.
\end{equation}
\section{MAIN RESULTS}
\label{sec:mainresults}
We begin with a result guaranteeing that  the space of threshold-type policies (see Definition \ref{arbitrarythresholds}) contains optimal update policies hence we can focus our attention to these policies for finding solutions to (\ref{minavaged}). 

Note that at time $t=Z_k$, the age $\Delta(t)$ is equal to $0$.  In the meanwhile, the battery level $E(t)$ will grow as more energy is harvested. In threshold policies, the threshold $\tau_{E(t)}$ changes according to the battery level $E(t)$ and a new sample is taken at the earliest time that the age $\Delta(t)$ exceeds the threshold $\tau_{E(t)}$. We define such policies as follows:
\begin{definition}
\label{arbitrarythresholds}
\normalfont
When $E(t) \in \lbrace\ell=1,...,B\rbrace$ represents the battery level at time $t$, an online policy is said to be a  \emph{threshold policy} if there exists $\tau_{\ell}$ for $\ell=1,...,B$ s.t.
\begin{equation}
\label{thresholdpolicy}
Z_{k+1} =\inf \left\lbrace t \geq Z_{k} : \Delta(t) \geq \tau_{E(t)} \right\rbrace ,
\end{equation}
\end{definition}
Note that a policy is said to be \emph{stationary} if its actions depend only on a current state while being independent of time.
An immediate observation is that given $\Delta(t)$ and $E(t)$ threshold policies do not depend on time, hence:
\begin{proposition}
All threshold policies  are stationary.
\end{proposition}
\begin{proof}
By definition, the update instants of a threshold policy only depend on the time elapsed since the last update, i.e., $\Delta(t)$, and the current battery level.
\end{proof}
We expect that such stationary policies can minimize $\bar{\Delta}$ among all online policies as energy arrivals follow a Poisson process which is memoryless. Due to the memorylessness of energy arrivals, the evolution of the system can be understood through a renewal type behaviour which suggests that an optimal policy should be stationary.

Indeed, we note the following as the first key result of this paper,
\begin{theorem}
\label{existopthreshold}
There exists a threshold policy that is optimal for solving (\ref{minavaged}).
\end{theorem}
\begin{proof}
See Appendix \ref{proof:existopthreshold}.
\end{proof}

 One significant challenge in the proof of Theorem \ref{existopthreshold} is that \eqref{minavaged} is an infinite time-horizon time-averaged MDP which has an uncountable state space. When the state space is countable, one can analyze infinite time-horizon time-averaged MDP by making a unichain assumption. However, this method cannot be directly applied when state space is uncountable. To resolve this, we use a modified version of the  ``vanishing discount factor" approach \cite{guohernandez2009} to prove Theorem \ref{existopthreshold} in two steps:

1. Show that for every $\alpha >0$, there exists a threshold policy that is optimal for solving 

$$\min_{\pi\in\Pi^{\mathsf{online}}}\mathbb {E}\left[ \displaystyle\int_{0}^{\infty}e^{-\alpha (t-a)} p(\Delta(t))  dt \right].$$

2. Prove that this property also holds when the discount factor $\alpha$ vanishes to zero.

In our search for an optimal policy, we can further reduce the space of policies:
\begin{definition}
\label{monotonicthresholds}
\normalfont
A threshold policy is said to be a \emph{monotone threshold} policy if $\tau_{1}\geq\tau_{2}\geq\ldots\geq\tau_B$. 
\end{definition}
Note that the definition of monotone threshold policies refers only to the case of thresholds that non-increasing in battery levels as opposed to the non-decreasing case.

Let $\Pi^{\rm{MT}}$ be the set of monotone threshold policies, then, the following is true: 

\begin{theorem}
\label{existopmothreshold}
There exists a monotone threshold policy $\pi \in \Pi^{\rm{MT}}$ that is optimal for solving (\ref{minavaged}).
\end{theorem}
\begin{proof}
See Appendix \ref{proof:existopmothreshold}.
\end{proof}
Theorem \ref{existopmothreshold} implies that in the optimal update policy, update packets are sent out more frequently when the battery level is high and less frequently when the battery level is low. This result is quite intuitive: If the battery is full, arrival energy cannot be harvested; if the battery is empty, update packets cannot be transmitted when needed and the age increases. Hence, both battery overflow and outage are harmful. Monotone threshold policies can address this issue. When the battery level $l$ is high, the threshold $\tau_l$ is small to reduce the chance of battery overflow; when the battery level $l$ is low, the threshold $\tau_l$ is high to avoid battery outage.  
\begin{figure}[htpb]
    \centering \includegraphics[scale=0.72]{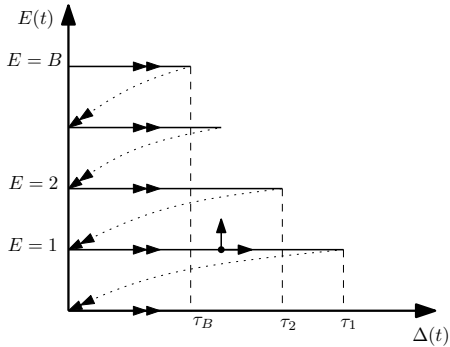}
\caption{An illustration of a monotone threshold policy.}
%\vspace{-0.099in}
\label{statespacew} 
\end{figure}

For a policy in $\Pi^{\rm{MT}}$, the state $(\Delta(t), E(t))$ does not spend a measurable amount of time anywhere $\Delta(t)\geq \tau_{E(t)}$ in which an update is sent out instantly reducing the battery level. Otherwise, the battery level is incremented upon energy harvests while the age is increasing linearly in time. The illustration in Fig. \ref{statespacew} shows the time evolution of the state $(\Delta(t), E(t))$ for policies in $\Pi^{\rm{MT}}$. If the energy level is $E(Z_{k})=j$ upon the previous update, then the inter-update time $X_{k+1} \in [\tau_m, \tau_{m-1}]$ holds if and only if  $m-j$ packets arrive during the inter-update time. In other words, reaching the battery state $m$ or higher is necessary and sufficient for the next inter-update duration being shorter than some $x$ when $x \in [\tau_{m},\tau_{m-1})$. Let $Y_{i}$ denote the duration required for $i\geq 1$ successive energy arrivals, which obeys the Erlang distribution at rate $\mu_H$ with parameter $i$,
\begin{equation}
\label{erlangdist}
P(Y_{i}\leq x)= 1- \displaystyle\sum_{v=0}^{i-1}\frac{1}{v!} e^{-\mu_H x}(\mu_H x)^{v},
\end{equation}

and let $Y_{i}=0$ for $i \leq 0$.

Accordingly, for  policies in $\Pi^{\rm{MT}}$, the cumulative distribution function (CDF) of inter-update durations,  can be expressed  as

\begin{eqnarray}
&\Pr(X_{k+1}\leq x \mid  E(Z_{k})=j)  = &\nonumber \\
&\begin{cases} 0, &\text{if~} x < \tau_{B}^, \\
\Pr(Y_{m-j}\leq x), &\text{if~}  \tau_{m}\leq x < \tau_{m-1}, \forall m \in \lbrace 2,...,B\rbrace, \\
\Pr(Y_{1-j}\leq x),  &\text{if~} \tau_{1}\leq x, 
\end{cases}&\nonumber\\ 
\label{cdferlang}
\end{eqnarray}

From (\ref{cdferlang}), an expression for the transition probability $\Pr(E(Z_{k+1})=i \mid  E(Z_{k})=j)$ for $i= 0,1, ...., B-1$ can be derived\footnote{Note that the event $E(Z_{k+1})=i$ happens if and only if $X_{k+1} \in [\tau_{i+1}, \tau_{i})$, accordingly $\Pr(E(Z_{k+1})=i \mid  E(Z_{k})=j)=\Pr(X_{k+1}\leq \tau_{i} \mid  E(Z_{k})=j)-\Pr(X_{k+1}\leq \tau_{i+1} \mid  E(Z_{k})=j)$.}
\begin{eqnarray}
&\Pr(E(Z_{k+1})=i \mid  E(Z_{k})=j)
=&\nonumber\\
&\begin{cases} \Pr(Y_{B-j}\leq \tau_{B-1}), &\text{if~} i=B-1, \\
\Pr(Y_{1+i-j}\leq \tau_{i})- \Pr(Y_{2+i-j}\leq \tau_{i+1}), &\text{if~}  i <B-1, \\
\end{cases}&\nonumber\\ 
\label{transitionerlang}
\end{eqnarray}

Hence, energy states sampled at update instants can be described as a Discrete Time Markov Chain (DTMC) with the transition probabilities in (\ref{transitionerlang}) (See Fig. \ref{DTMCw}). When thresholds are finite, this DTMC is ergodic as any energy state is reachable from any other energy state in $B-1$ steps with positive probability. 

\begin{figure}[htpb]
    \centering \includegraphics[scale=0.48]{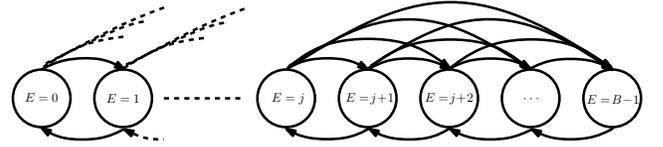}
\caption{The DTMC for energy states sampled at update times.}
%\vspace{-0.099in}
\label{DTMCw} 
\end{figure}

 Any optimal policy in $\Pi^{\rm{MT}}$ has the following property:
\begin{theorem}
\label{fixedpthreshold}
%An optimal policy for solving (\ref{minavaged}) is a monotone threshold policy  that satisfies the following 
%\begin{itemize}
%\item  $\tau^*_B\leq \ldots\leq\tau^*_{2}\leq\tau^*_{1}$;
%\item The threshold $\tau^*_B$ for sending an update packet when the battery is full is equal to the minimum time-average expected age, i.e.,
%\begin{equation}
%\tau^*_B = \bar{\Delta}_{\pi^*} =  \min_{\pi\in\Pi^{\mathsf{online}}} \bar{\Delta}_\pi. 
%\end{equation}
%\end{itemize}
An optimal policy for solving (\ref{minavaged}) is a monotone threshold policy  that satisfies the following 
\begin{equation}
p(\tau^*_B) = \bar{p}_{\pi^*} =  \min_{\pi\in\Pi^{\mathsf{online}}} \bar{p}_\pi. 
\end{equation}
where $\pi^*$ is a monotone threshold policy solving (\ref{minavaged}) and $\tau^*_B$ is its age threshold for the full battery case.
\end{theorem}
\begin{proof}
See Appendix \ref{proof:fixedpthreshold}.
\end{proof}

The result in Theorem \ref{fixedpthreshold} exhibits a structural property of optimal policies which also appears in the sampling problem that was studied in \cite{DBLP:journals/corr/abs-1812-07241} . The sampling problem in \cite{DBLP:journals/corr/abs-1812-07241} considered sources without energy harvesting, where the packet transmission times were \emph{i.i.d.} and non-zero.  On the one hand, the optimal sampling policy in Theorem 1 of \cite{DBLP:journals/corr/abs-1812-07241} is a threshold policy on an expected age penalty term, and the threshold is exactly equal to the optimal objective value. On the other hand, a sampling problem for an energy harvesting source with zero packet transmission time is considered in the current paper. The optimal sampling policy in  Theorem \ref{fixedpthreshold} can be rewritten as  
\[
Z_{k+1} =\inf \left\lbrace t \geq Z_{k} : p(\Delta(t)) \geq p(\tau_{E(t)}^{*}) \right\rbrace 
\]
which is a multi-threshold policy on the age penalty function, each threshold  $p(\tau_{\ell}^{*})$ corresponding to a battery level $\ell$. Further, the threshold $p(\tau_{B}^{*})$ associated with a full battery size $E(t)=B$  is equal to the optimal objective value. The results in these two studies are similar to each other. Together, they provide a unified view on optimal sampler design for sources both with and without energy harvesting capability. The proof techniques in these two studies are of fundamental difference.  

\subsection{Average Age Case}
If we take the age-penalty function as an identity function, i.e., $p(\Delta)=\Delta$, then (\ref{minavaged}) becomes the problem of minimizing the time-average expected age. In this case, the result in Theorem \ref{fixedpthreshold} implies that in optimal monotone threshold policies, inter-update durations can be small as much as the minimum average AoI only when the battery is full. From results in
~\cite{TanISIT2017} and ~\cite{WuYang2017}, we know that the minimum average AoI for the infinite battery case is $\frac{1}{2\mu_{H}}$ and this can be achieved asymptotically using the best-effort scheme in ~\cite{WuYang2017} or with a threshold policy ~\cite{TanISIT2017} where all thresholds are nearly equal to $\frac{1}{\mu_{H}}$. On the other hand, according to  Theorem \ref{fixedpthreshold}, the optimal threshold for the full battery level tends to $\frac{1}{2\mu_{H}}$ as the battery capacity increases. This shows that the optimal monotone threshold policies remain structurally dissimilar to asymptotically optimal policies when the battery capacity is approaching to infinity. The result is more useful when the battery capacity is finite as it may lead to the optimal threshold values of the other battery levels. We will use this in  an algorithm for finding near optimal policies for any given integer sized battery capacity.      
In addition, the special case of Theorem \ref{fixedpthreshold} for average age \cite{8437573} can be derived from a more general result which we  provide in Lemma \ref{diffmoments}. This result shows a relation between the partial derivatives of a non-negative random variable with respect to the thresholds determining the random variable in a similar way to the inter-update duration case.   
\begin{lemma}
\label{diffmoments}
Suppose $X$ is a r.v. that satisfies the following:
\begin{eqnarray*}
\Pr(X\leq x)  = 
\begin{cases} 0 &\text{if~} x < \tau_{B}, \\
F_{i}(x) &\text{if~}  \tau_{i}\leq x < \tau_{i-1}, \forall i \in \lbrace 2,...,B \rbrace , \\
F_{1}(x)  &\text{if~} \tau_{1}\leq x, 
\end{cases}&
\end{eqnarray*}
where $0<\tau_{B}\leq ... \leq \tau_{2} \leq \tau_{1}$ and for each $i \in \lbrace 1,...,B \rbrace$  $F_{i}(x)$ is the CDF of a non-negative random variable. Then:
\begin{equation*}
\frac{\partial}{\partial \tau_{i}}\mathbb {E}\left[ X^{2} \right]=2\tau_{i}\frac{\partial}{\partial \tau_{i}}\mathbb {E}\left[ X \right].
\end{equation*}
\end{lemma}
\begin{proof}
See Appendix \ref{proof:diffmoments}.
\end{proof}
\begin{corollary}
The inter-update intervals, $X$, for any $\pi \in \Pi^{\rm{MT}}$ satisfy the following: 
\begin{equation}
\frac{\partial}{\partial \tau_{i}}\mathbb {E}\left[ X^{2} \mid E=j\right]=2\tau_{i}\frac{\partial}{\partial \tau_{i}}\mathbb {E}\left[ X \mid E=j\right],
\end{equation}
$\forall(i,j)\in \lbrace 1, 2,...,B\rbrace^{2}$ where $\mathbb {E}\left[ X \mid E=j\right] \triangleq \mathbb {E}\left[ X_{k} \mid E(Z_{k})=j\right]$ and $\mathbb {E}\left[ X^{2} \mid E=j\right] \triangleq \mathbb {E}\left[ X_{k}^{2} \mid E(Z_{k})=j\right]$ .
\end{corollary}
Note that the transition probabilities (\ref{transitionerlang}) do not depend on $\tau_{B}$ hence the steady-state probabilities obtained from (\ref{transitionerlang}) also do not depend on $\tau_{B}$. This leads to a property of $\tau_{B}$ the average age case of Theorem \ref{fixedpthreshold} as shown in \cite{8437573}. 
The unit-battery case , i.e., $B=1$ case was solved in ~\cite{TanISIT2017} and  ~\cite{WuYang2017}. For completeness, this result is summarized in Theorem \ref{B1age}.
\begin{theorem}
\label{B1age}
When $B=1$, the average age $\bar{\Delta}$ can be expressed as
\begin{equation}
\label{B1avage}
\bar{\Delta}=\frac{\frac{1}{2}(\mu_{H}\tau_{1})^{2}+ e^{-\mu_{H}\tau_{1}}\!(\mu_{H}\tau_{1}+1)}{\mu_{H}(\mu_{H}\tau_{1}+e^{-\mu_{H}\tau_{1}})},
\end{equation}
and $\tau_{1}^{*}=\bar{\Delta}_{\pi^*}=\frac{1}{\mu_{H}}2W(\frac{1}{\sqrt{2}})$ where $W(\cdot)$ is the Lambert-W function.

\end{theorem} 
\begin{proof}
See Appendix \ref{proof:B1age}.
\end{proof}
\begin{theorem}
\label{B2age}
When $B=2$, the average age $\bar{\Delta}$ can be expressed as:
\begin{eqnarray}
&\bar{\Delta}=& \nonumber\\
&\frac{\!\frac{\alpha_{2}^{2}}{2}\!+\! e^{-\alpha_{2}}\![\alpha_{2}+1+\rho_{1}(\alpha_{2}^{2}+2\alpha_{2}+2)]\! -e^{-\alpha_{1}}\![\alpha_{1}+1+\rho_{1}(\alpha_{1}^{2}+\alpha_{1}+1)]}{\mu_{H}\left( \alpha_{2}+e^{-\alpha_{2}}[1+\rho_{1}(\alpha_{2}+1)]-e^{-\alpha_{1}}[1+\rho_{1}\alpha_{1}]\right) },&
\label{eq:B2age}
\end{eqnarray}
where
\[
\rho_{1}=\frac{e^{-\alpha_{1}}}{1-e^{-\alpha_{1}}\alpha_{1}},
\]
and
\[
\alpha_{1}=\mu_{H}\tau_{1}, \alpha_{2}=\mu_{H}\tau_{2}.
\]

\end{theorem}
\begin{proof}
See Appendix \ref{proof:B2age}.
\end{proof}
\subsection{An Algorithm for Finding Near Optimal Policies}
\label{sec:AoIop}
We propose an algorithm to find a near optimal policy $\pi \in \Pi^{\rm{MT}}$ such that $\bar{\Delta}_{\pi}-\bar{\Delta}_{\pi^*}\leq \frac{1}{2^{q+1}\mu_H}$ for any given $B$ and $q \in \mathbb{Z}^{+}$. Let $m_{1}(\tau_{1},\tau_{2},...,\tau_{B})$ and $m_{2}(\tau_{1},\tau_{2},...,\tau_{B})$ denote the functions such that:
\begin{equation}
m_{1}(\tau_{1},\tau_{2},...,\tau_{B})=\sum_{j=0}^{B-1}\mathbb {E}\left[ X \mid E=j\right] \Pr(E=j),
\label{monotoneXa}
\end{equation}
\begin{equation}
m_{2}(\tau_{1},\tau_{2},...,\tau_{B})=\sum_{j=0}^{B-1} \mathbb {E}\left[ X^{2} \mid E=j\right] \Pr(E=j),
\label{monotoneCa}
\end{equation}
where  $\Pr(E=j)$ is the steady-state probability for energy state $j$, $\mathbb {E}\left[ X \mid E=j\right] \triangleq \mathbb {E}\left[ X_{k} \mid E(Z_{k})=j\right]$ and $\mathbb {E}\left[ X^{2} \mid E=j\right] \triangleq \mathbb {E}\left[ X_{k}^{2} \mid E(Z_{k})=j\right]$.

Note that it is straight forward to derive $m_{1}(\tau_{1},\tau_{2},...,\tau_{B})$ and $m_{2}(\tau_{1},\tau_{2},...,\tau_{B})$ using (\ref{cdferlang}) and (\ref{transitionerlang}), hence we assume these functions are available for any $B$. 

In the below theorem , we state the main result that we will use in an algorithm for finding near optimal policies:
\begin{theorem}
\label{momentsthreshold}
For $B>1$, the equation 
\begin{equation}
2\tau_{B}m_{1}(\tau_{1},\tau_{2},...,\tau_{B})-m_{2}(\tau_{1},\tau_{2},...,\tau_{B})=0,
\label{transcendentalage}
\end{equation}
has a solution with monotone non-increasing thresholds, i.e., $\tau_{B}\leq ... \leq \tau_{2} \leq \tau_{1}$ if and only if $\tau_{B}\geq \bar{\Delta}_{\pi^*}$.
\end{theorem}

Algorithm \ref{alg1} uses this result to find a near optimal policy $\pi \in \Pi^{\rm{MT}}$ such that $\bar{\Delta}_{\pi}-\bar{\Delta}_{\pi^*}\leq \frac{1}{2^{q+1}\mu_H}$. Each iteration in Algorithm \ref{alg1} halves the interval where the minimum average AoI can be found based on the existence of solution to (\ref{momentsthreshold}) with the current estimate of the smallest threshold $\hat{\tau}_{B}$.  Accordingly, it is guaranteed that Algorithm \ref{alg1} finds a solution within a gap to the optimal value that is $\frac{1}{2^{q+1}\mu_H}$.  

 Algorithm \ref{alg1} assumes a numerical solver that can solve the transcendental equation in (\ref{transcendentalage}), however, the exact solution is required only once at the final step while iterations only require verifying the existence of a solution to (\ref{momentsthreshold}).
\begin{algorithm} % enter the algorithm environment
\caption{Find $\pi \in \Pi^{\rm{MT}}$ such that $\bar{\Delta}_{\pi}-\bar{\Delta}_{\pi^*}\leq \frac{1}{2^{q+1}\mu_H}$} % give the algorithm a caption
\label{alg1} % and a label for \ref{} commands later in the document
\begin{algorithmic} % enter the algorithmic environment
    \REQUIRE $B \geq 1 \wedge q \geq 1$
    \ENSURE $\bar{\Delta}_{\pi}-\bar{\Delta}_{\pi^*}\leq \frac{1}{2^{q+1}\mu_H}$
    \STATE $\tau_{B}^{-}\leftarrow \frac{1}{2\mu_{H}}$, $\tau_{B}^{+}\leftarrow \frac{1}{\mu_{H}}$
   
    \FOR{$i=1,2,...,q$}
    	\STATE $\hat{\tau}_{B}\leftarrow\frac{\tau_{B}^{-}+\tau_{B}^{+}}{2}$
        \IF{$\exists \tau_{B-1}\leq ... \leq \tau_{2} \leq \tau_{1}$ s.t. $\tau_{B-1}\geq \hat{\tau}_{B}$ and\\
        $2\hat{\tau}_{B}m_{1}(\tau_{1},\tau_{2},...,\hat{\tau}_{B})-m_{2}(\tau_{1},\tau_{2},...,\hat{\tau}_{B})=0$}
        
            \STATE $\tau_{B}^{+}\leftarrow \hat{\tau}_{B}$
            
        \ELSE
            \STATE $\tau_{B}^{-}\leftarrow \hat{\tau}_{B}$
        \ENDIF
    \ENDFOR 
    \STATE Solve $2\hat{\tau}_{B}m_{1}(\tau_{1},\tau_{2},...,\hat{\tau}_{B})-m_{2}(\tau_{1},\tau_{2},...,\hat{\tau}_{B})=0$
    \RETURN $\pi=(\tau_1,\tau_{2},...,\hat{\tau}_{B})$ 
\end{algorithmic}
\end{algorithm}
\section{NUMERICAL RESULTS}
\label{sec:nm}
For battery sizes $B=1,2,3,4$, the policies in $\Pi^{\rm{MT}}$ are numerically optimized giving AoI versus energy arrival rate (Poisson) curves in Fig \ref{AgeEnergy}. We give the corresponding threshold values in Table \ref{thresholdtable}. These results were obtained through exhaustive search for possible threshold values, and Monte Carlo analysis for approximating AoI values in the simulation of the considered system and policies without relying on analytical results. It can be seen that these optimal thresholds and corresponding AoI values (in Table \ref{thresholdtable}) validate Theorem \ref{fixedpthreshold}.  
Fig. \ref{AgeB2tau1} and \ref{AgeB2tau2} show the dependency of AoI on threshold values $\tau_{1}$  and $\tau_{2}$ which is consistent with the result in Theorem \ref{B2age} for the special case of $B=2$.
\begin{figure}[htpb]
\centering
  \begin{psfrags}
    \psfrag{Y}[t]{$\bar{\Delta}$}
    \psfrag{X}[b]{$\mu_{H}$}
    \psfrag{a}[l]{$B=1$}
    \psfrag{b}[l]{$B=2$}
    \psfrag{c}[l]{$B=3$}
    \psfrag{d}[l]{$B=4$}
    \psfrag{e}[l]{$B=5$}
    \psfrag{f}[l]{$B=\infty$}
   \includegraphics[scale=0.19]{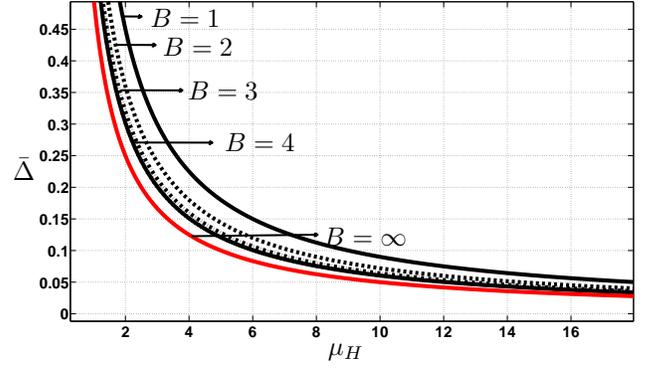}
    \end{psfrags}
\caption{AoI versus energy arrival rate (Poisson) for different battery sizes $B=1,2,3,4$.}
\label{AgeEnergy} 
\end{figure}

\begin{figure}[htpb]
\centering
  \begin{psfrags}
    \psfrag{Y}[t]{$\bar{\Delta}$}
    \psfrag{X}[b]{$\tau_{1}$}
   \includegraphics[scale=0.19]{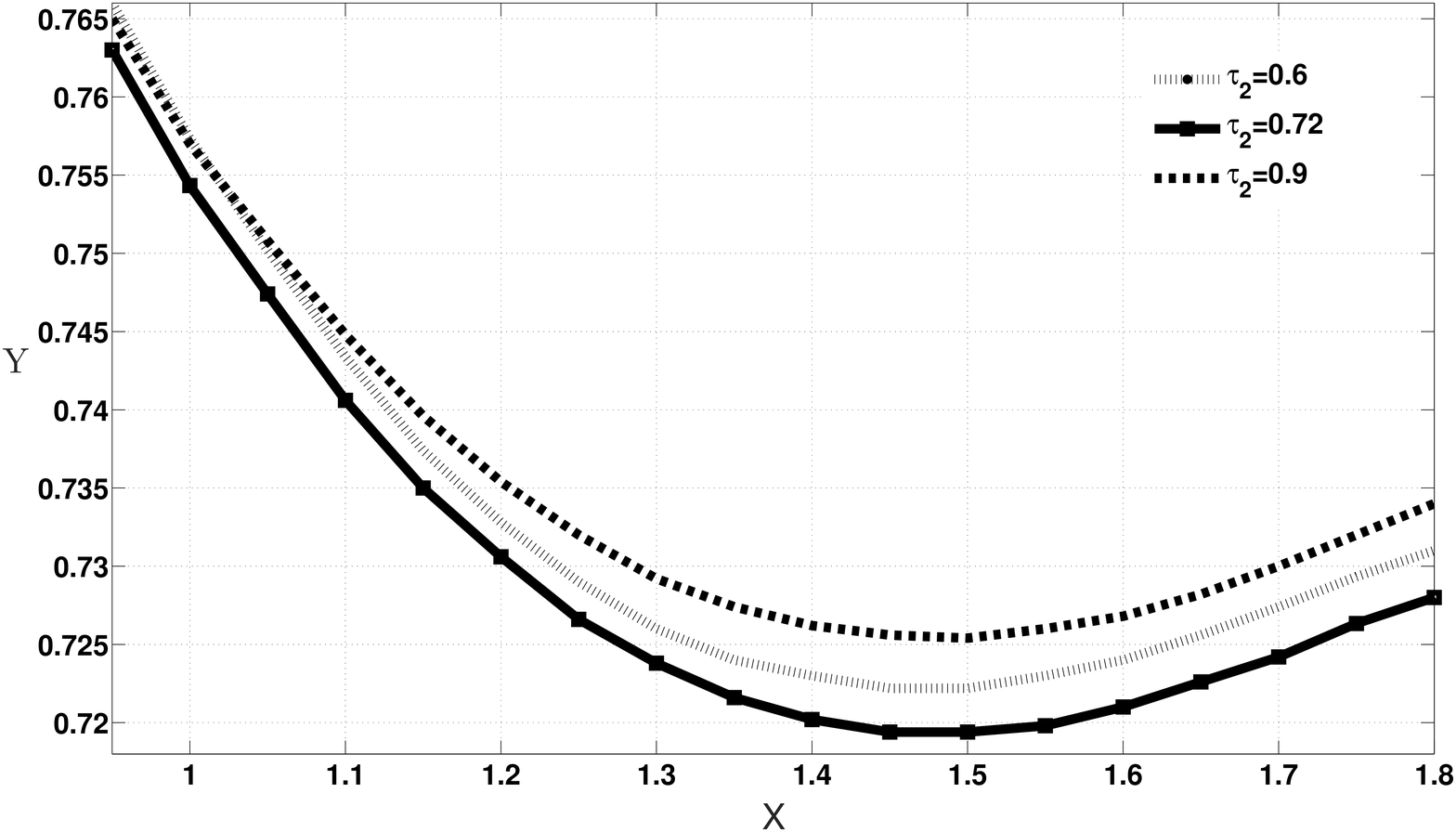}
    \end{psfrags}
\caption{AoI versus $\tau_{1}$ against various $\tau_{2}$ values for $B=2$ and $\mu_{H}=1$.}
\label{AgeB2tau1} 
\end{figure}

\begin{figure}[htpb]
\centering
  \begin{psfrags}
    \psfrag{Y}[t]{$\bar{\Delta}$}
    \psfrag{X}[b]{$\tau_{2}$}
  \includegraphics[scale=0.19]{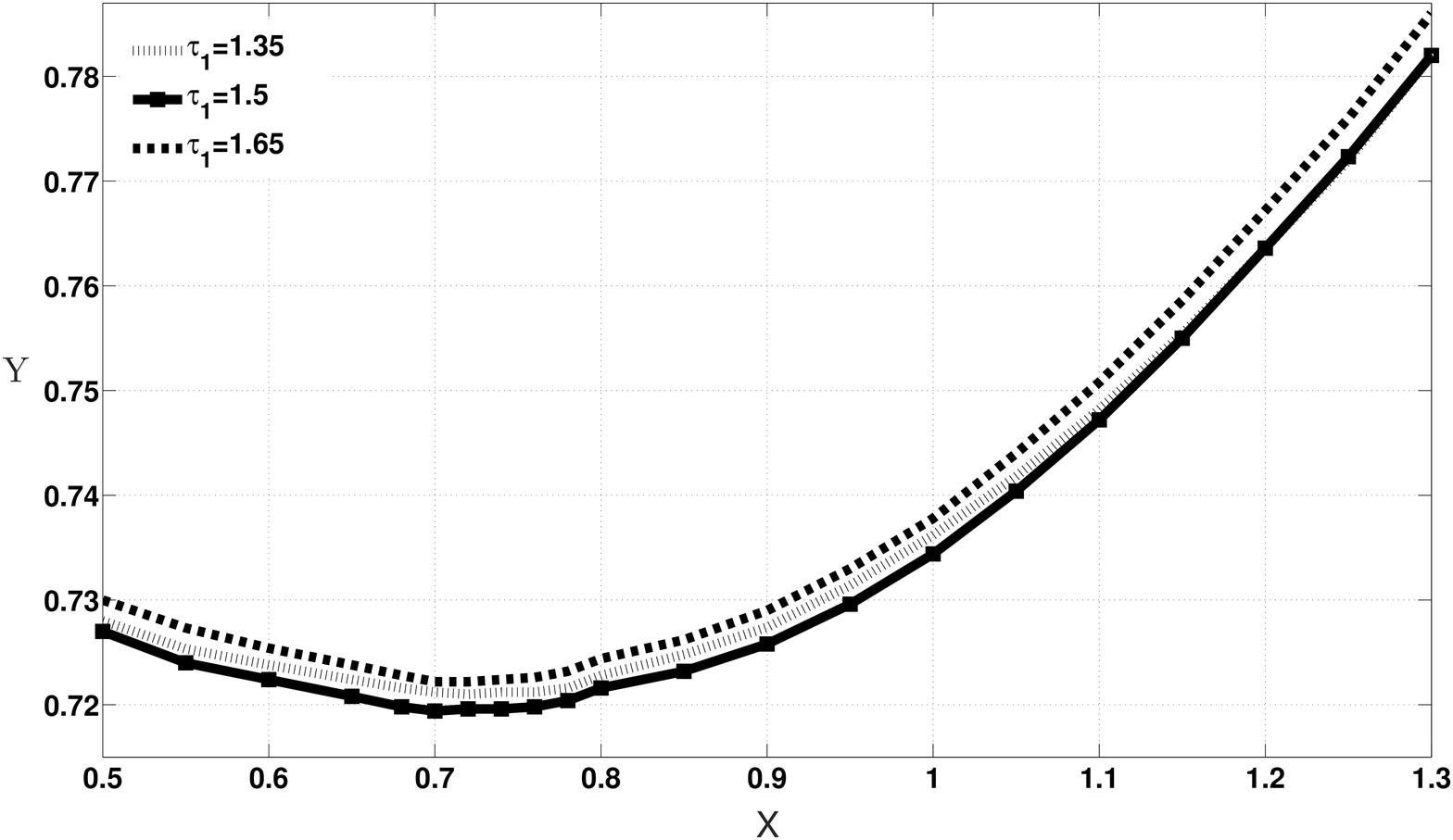}
    \end{psfrags}
\caption{AoI versus $\tau_{2}$ against various $\tau_{1}$ values for $B=2$ and $\mu_{H}=1$.}
\label{AgeB2tau2} 
\end{figure}

\begin{table}[]
\centering
\caption{Optimal thresholds for different battery sizes for $\mu_{H}=1$}
\label{thresholdtable}
\begin{tabular}{|l|l|l|l|l|l|l|}
\hline
    & $\tau_{1}$ & $\tau_{2}$     & $\tau_{3}$     & $\tau_{4}$       &$\bar{\Delta}_{\pi^*} $    \\ \hline
$B=1$ & 0.90         & - &  -    &   -      & 0.90 \\ \hline    
$B=2$ & 1.5         & 0.72 &  -    &   -      & 0.72 \\ \hline
$B=3$ & 1.5         & 1.2  & 0.64 &   -     & 0.64\\ \hline
$B=4$ & 1.5         & 1.2  & 0.86 & 0.604    & 0.604\\ \hline
%$B=5$ & 1.5         & 1.2  & 0.96 & 0.9   & 0.582  & 0.582\\ \hline
\end{tabular}
\end{table}

\vspace{-0.1 in}
\section{CONCLUSION}
\label{sec:conc}
We have studied optimizing a non-linear age penalty in the generation  
and transmission of status updates by an energy harvesting source with  
a finite battery. An optimal status updating policy for
minimizing the time-average expectation of a general non-decreasing age
function $p(\cdot)$ has been obtained. The policy
has a monotonic threshold structure: (i) each new update is sent out only when the age is higher than a threshold and (ii) the threshold is a
non-increasing function of the instantaneous battery level such that the
updates are sent out more frequently when the battery level is high.
Furthermore, we have identified an interesting relationship between the
smallest optimal threshold $\tau_{B}^{*}$ (i.e., the threshold corresponding to a full
battery level) and the optimal objective value $\bar{p}_{\pi^{*}}$ (i.e., the
minimum achievable time-average expected age penalty), which is given by
\begin{align}
p(\tau_{B}^{*}) = \bar{p}_{\pi^{*}}.
  \end{align}  
\vspace{-0.1 in}
\bibliography{AgeOfInformation}

\appendix
\subsection{The Proof of Theorem \ref{existopthreshold}}
\label{proof:existopthreshold}

In order to prove Theorem \ref{existopthreshold}, we use a modified version of the ``vanishing discount factor" approach \cite{guohernandez2009} which consists of 2 steps:

\emph{Step 1}. Show that for every $\alpha >0$, there exists a threshold policy that is optimal for solving 

$$\min_{\pi\in\Pi^{\mathsf{online}}}\mathbb {E}\left[ \displaystyle\int_{0}^{\infty}e^{-\alpha t} p(\Delta(t))  dt \right].$$

\emph{Step 2}. Prove that this property still holds when the discount factor $\alpha$ vanishes to zero. 

We first discuss \emph{Step 1}. Recall that $\mathcal{F}_{t}$ represents the information about the energy arrivals and the update policy during $[0,t]$. Given $\mathcal{F}_{a}$, we are interested in finding the optimal online policy during $[a,\infty)$, which is formulated as

\begin{equation}
\label{adiscounted}
\displaystyle\min_{\pi\in\Pi^{\mathsf{online}}} \mathbb {E}\left[ \displaystyle\int_{a}^{\infty}e^{-\alpha (t-a)}p(\Delta(t))  dt \middle| \mathcal{F}_{a}\right].
\end{equation}

Observe that, in (\ref{adiscounted}), the term $e^{-\alpha (t-a)}$ ensures that the exponential decay always starts from unity so that the problem is independent of $a$ given $\mathcal{F}_{a}$. In addition, this problem has the following nice property:
\begin{lemma}
\label{starcost}
There exists an optimal solution to (\ref{adiscounted}) that depends on $\mathcal{F}_{a}$ only through $(\Delta(a),E(a))$. That is, $(\Delta(a),E(a))$ is a sufficient statistic for solving (\ref{adiscounted}).
\end{lemma}

\begin{proof}
In Problem \eqref{adiscounted}, the age evolution $\{\Delta(t), t \geq a\}$ is determined by the initial age $\Delta(a)$ at time $a$ and the update policy during $[a,\infty)$. Further, the update policy during $[a,\infty)$ is determined by the initial age $\Delta(a)$, the initial battery level $E(a)$, and the energy  counting process $\{H(t)-H(a), t\geq a\}$. Hence, $\{\Delta(t), t \geq a\}$ is determined by $\Delta(a)$, $E(a)$, and $\{H(t)-H(a), t\geq a\}$.

Recall that $\Delta(0)$ and $E(0)$ are fixed. Hence, for any online update policy, the online update decisions during $[0,a]$ depends only on $\{H(t), t\leq a\}$. Hence, $\mathcal{F}_{a}$ is determined by $\{H(t), t\leq a\}$. Because $\{H(t), t\geq 0\}$ is a compound Poisson process, $\{H(t)-H(a), t\geq a\}$ is independent of $\{H(t), t\leq a\}$. Hence, $\{\Delta(t), t \geq a\}$ depends on $\mathcal{F}_{a}$ only through $\Delta(a)$ and  $E(a)$. By this, $(\Delta(a),E(a))$ is a sufficient statistic for solving (\ref{adiscounted}).\end{proof}

 By using Lemma \ref{starcost}, we can simplify (\ref{adiscounted}) as \eqref{Jdefinition} and define a cost function $J_{\alpha}(\Delta(a),E(a))$ which is the optimal objective value of (\ref{Jdefinition}):
\begin{align}
\label{Jdefinition}
&J_{\alpha}(\Delta(a),E(a)) :=\!\!\!\! \displaystyle\min_{\pi\in\Pi^{\mathsf{online}}} \mathbb {E}\left[ \displaystyle\int_{a}^{\infty}\!\!\!\!e^{-\alpha (t-a)}p(\Delta(t))  dt \middle| \mathcal{F}_{a}\right]=& \nonumber\\
&\!\!\!\! \displaystyle\min_{\pi\in\Pi^{\mathsf{online}}} \mathbb {E}\left[ \displaystyle\int_{a}^{\infty}\!\!\!\!e^{-\alpha (t-a)}p(\Delta(t))  dt \middle| \Delta(a),E(a)\right]& 
\end{align}

Furthermore, one important question is: Given that the previous update occurs at $Z_k = a$, how to choose the next update time $Z_{k+1}$. This can be formulated as 
\begin{eqnarray}
\label{kfinite}
&\displaystyle\min_{(Z_1, \ldots, Z_k=a, Z_{k+1},\ldots)\in\Pi^{\mathsf{online}}}&\nonumber\\
& \mathbb {E}\left[ \displaystyle\int_{a}^{\infty}e^{-\alpha (t-a)}p(\Delta(t))  dt \middle| Z_k = a, \Delta(a) = 0,E(a)\right]&,
\end{eqnarray}
where we have used the fact that if $Z_k = a$, then $\Delta(a)=\Delta(Z_k)=0$.

According to the definition of $\Pi^{\mathsf{online}}$, $Z_{k+1}$ is a finite Markov time, i.e., stopping time, hence the problem of finding $Z_{k+1}$ for a solution to (\ref{kfinite}) can be formulated as an infinite horizon optimal stopping problem in the interval $[a,\infty)$.  We will consider a \emph{gain} \cite{peskir2006optimal} process $G=(G_{t})_{t \geq a}$ adapted to the filtration $\mathcal{F}_{t}$ where  a stopping time $Z_{k+1}$ for a solution to (\ref{kfinite}) maximizes $\mathbb {E}\left[ G_{Z_{k+1}} \mid \mathcal{F}_{a}\right]$ when we choose $Z_{k+1}$ from a family of stopping times based on $\mathcal{F}_{t}$. Let $\mathfrak{M}_{a}$ denote this family of $Z_{k+1}$s which can be expressed as:
\begin{align*}
&\mathfrak{M}_{a}= \nonumber\\
&\left\lbrace  Z_{k+1} \geq a : \Pr(Z_{k+1}<\infty)=1,\left\lbrace Z_{k+1} \leq t\right\rbrace \in \mathcal{F}_{t}, \forall t \geq a\right\rbrace.
\end{align*} 

Note that a stopping time in $\mathfrak{M}_{a}$ may violate energy causality however our definition of the gain process will guarantee that those stopping times cannot be optimal. 

We will define the gain process $(G_{t})_{t \geq a}$ based on the value of the discounted cost when an update is sent at a particular time $t$.  The gain process $(G_{t})_{t \geq a}$ for $E(t)>0$ corresponds to the additive inverse of this cost and can be written  as follows:
 \begin{align}
\label{gainalt}
&G_{t}=-\displaystyle\min_{\pi\in\Pi^{\mathsf{online}}}\nonumber\\ 
&\mathbb {E}\left[ \displaystyle\int_{a}^{\infty}e^{-\alpha (w-a) }p(\Delta(w))  dw \middle|Z_{k}=a, Z_{k+1}=t, E(t)\right],&\nonumber\\ & E(t)> 0.
\end{align} 

Note that the stopping time cannot be at time $t$ when $E(t)=0$ as there is no energy to send another update in that case. To cover this case, we set $G_{t}$ to $-\infty$ so that a stopping time $Z_{k+1}$ maximizing $\mathbb {E}\left[ G_{Z_{k+1}} \mid \mathcal{F}_{a}\right]$ should satisfy energy causality hence belongs to an online policy. In other words, the stopping time $Z_{k+1}$ in a solution to (\ref{kfinite}) maximizes $\mathbb {E}\left[ G_{Z_{k+1}} \mid \mathcal{F}_{a}\right]$ among all the stopping times in $\mathfrak{M}_{a}$. 

Alternatively, the gain process  $(G_{t})_{t \geq a}$ can be  expressed in terms of the cost defined in (\ref{Jdefinition}) as follows 
\begin{align}
\label{ageJ}
&G_{t} =-\int_{a}^{t}e^{-\alpha (w-a) }p(w-a) dw-\nonumber\\&\mathbb {E}\left[ \int_{t}^{\infty}\!\!e^{-\alpha(w-a)}p(\Delta(w)) dw\middle|Z_{k}=a, Z_{k+1}=t, E(t)\right]\nonumber\\
&=-\int_{a}^{t}e^{-\alpha(w-a)}p(w-a) dw  -e^{-\alpha(t-a)}J_{\alpha}(0, E(t)-1),
\end{align}
for $t \geq a$ and $E(t)>0$.

Let's define $J(0,-1):=\infty$ so that (\ref{ageJ}) holds for the $E(t)=0$ as well. Notice that the process $G_{t}$ is driven by the random process $E(t)$ which is not conditioned on any particular value of $E(a)$ while being adapted to the filtration $\mathcal{F}_{t}$. However, for a policy solving (\ref{kfinite}),  the stopping time $Z_{k+1}$ depends on $E(a)$ as it maximizes $\mathbb {E}\left[ G_{Z_{k+1}} \mid \mathcal{F}_{a}\right]$ which depends on $E(a)$ through the filtration $\mathcal{F}_{a}$.    

Accordingly, we define the stopping problem of maximizing the expected gain in the given interval $[a,\infty)$ as in the following:
\begin{equation}
\label{stoppingae}
\displaystyle\max_{t \in \mathfrak{M}_{a}} \mathbb {E}\left[ G_{t} \mid \mathcal{F}_{a}\right].
\end{equation}
Based on this formulation, we will show that the optimal stopping time exists and is given by the
following stopping rule for $Z_{k+1}$:
\begin{equation}
\label{rulesnell}
Z_{k+1} = \inf \lbrace t \geq Z_{k}=a : G_{t} = S_{t} \rbrace,
\end{equation}
where $S$ is the Snell envelope \cite{peskir2006optimal} for $G$:
\begin{equation}
\label{snellsw}
S_{t}= \esssup_{t' \in \mathfrak{M}_{t}}  \mathbb {E}\left[ G_{t'} \mid\mathcal{F}_{t}\right].
\end{equation}

Showing that $Z_{k+1}$ in (\ref{rulesnell}) is finite w.p.1 is sufficient to prove the existence of the optimal stopping time and the optimality of the stopping rule in (\ref{rulesnell})(see \cite[Theorem 2.2.]{peskir2006optimal}). Consider the lemma below and its proof in order to see the finiteness of $Z_{k+1}$ in (\ref{rulesnell}):
\begin{lemma}
\label{stoppingexists}
For the stopping rule in (\ref{rulesnell}) $Z_{k+1}$ is finite w.p.1, i.e.,  $\Pr(Z_{k+1}<\infty)=1$.
%There is an optimal stopping time $Z_{k}$ as in (\ref{rulesnell}).
\end{lemma}
\begin{proof}
Consider the Markov time $Q_{k+1}$ which is defined as follows:
\begin{equation}
Q_{k+1} := \inf \lbrace t \geq Z_{k}=a : E(t)=B, G_{t} = S_{t} \rbrace.
\end{equation}
Clearly, the stopping time $Z_{k+1}$ chosen in (\ref{rulesnell}) is earlier than $Q_{k+1}$ as $Q_{k+1}$ has an additional stopping condition $E(t)=B$. This means that if $\Pr(Q_{k+1}<\infty)=1$, then $\Pr(Z_{k+1}<\infty)=1$.

Accordingly, for the proof of this lemma, it is sufficient to show that $Q_{k+1}$ is finite w.p.1. We will show this by showing the finiteness of (i) the first time $t \geq Z_{k}=a$ such that $E(t)=B$,  and (ii) the duration between this time and the Markov time $Q_{k+1}$. Note that $E(t)=B$ condition is always satisfied after it reached for the first time. Let $R_{k+1}$ be the Markov time representing the first time when $E(t)=B$ is satisfied:
\begin{equation}
R_{k+1} := \inf \lbrace t \geq Z_{k}=a : E(t)=B \rbrace.
\end{equation}

(i) Observe that the Markov time $R_{k+1}$ is finite w.p.1 as it is stochastically dominated by $a+ Y_{B}$ where $Y_{B}$ is an Erlang distributed random variable with parameter $B$ which obeys (\ref{erlangdist}) and $\Pr(Y_{B}< \infty)=1$. 

(ii) In order to see that $Q_{k+1}-R_{k+1}$ is also finite, consider the time period after $R_{k+1}$, i.e., $[R_{k+1},\infty)$. As $E(t)=B$ for any $t\geq R_{k+1}$, the evolution of $G_{t}$ becomes deterministic after $t \geq R_{k+1}$:
\begin{eqnarray}
\label{gaincost} 
&G_{t} =& \nonumber\\
&-\int_{a}^{t}e^{-\alpha(w-a)}p(w-a) dw -e^{-\alpha(t-a)}J_{\alpha}(0, B-1)&,
\end{eqnarray}
for $t \geq R_{k+1}$.

On the other hand, for $t\geq  R_{k+1}$, the Snell envelope is $S_{t}= \esssup_{t' \in \mathfrak{M}_{t}}G_{t'}= \sup_{t'\geq t}  G_{t'}$. We will show that $G_{t}$ is always non-increasing after some finite time so that $S_{t}=G_{t}$ is always satisfied after that time.  

In order to see this, consider the change in $G_{t}$ for $t\geq R_{k+1}$. As
\begin{eqnarray}
&-\frac{\partial}{\partial t}\left[ \int_{a}^{t}e^{-\alpha(w-a)}p(w-a) dw +e^{-\alpha(t-a)}J_{\alpha}(0, B-1)\right] =&\nonumber\\ &e^{-\alpha(t-a)}\left(\alpha J_{\alpha}(0, B-1)- p(t-a)\right),&\nonumber\\
\end{eqnarray}
and $p(t-a)$ is non-decreasing, for $t\geq R_{k+1}$, $G_{t}$ is non-increasing if $t\geq t_{c}$ for  some $t_{c}$ such that 
\begin{equation}
\label{tcdef}
t_{c}:=\inf\lbrace t \geq a: p(t-a)=\alpha J_{\alpha}(0, B-1)\rbrace.
\end{equation}
 This implies that, for $t\geq \max \lbrace R_{k+1}, t_{c}\rbrace$, $G_{t}= \sup_{t'\geq t} G_{t'}$  and  hence $S_{t}=G_{t}$. Accordingly, the stopping conditions of $Q_{k+1}$ are satisfied for the first time when $t = \max \lbrace R_{k+1}, t_{c}\rbrace$ which means $Q_{k+1}=\max \lbrace R_{k+1}, t_{c}\rbrace$. 

As $\alpha J_{\alpha}(0, B-1)$ is finite, 
 $t_{c}$ is finite which implies $Q_{k+1}$ is finite w.p.1 as $R_{k+1}$ is finite w.p.1. This completes the proof.
\end{proof}
We just showed that the Markov time in (\ref{rulesnell}) is finite w.p.1 and this means that it is the optimal stopping time by \cite[Theorem 2.2.]{peskir2006optimal}. Next, we show that the optimal stopping rule in (\ref{rulesnell}) is a threshold policy by using the properties of the cost function in (\ref{Jdefinition}).  To relate the optimal stopping time and the cost function in (\ref{Jdefinition}), we will express the Snell envelope in an alternative way.

Notice that the Snell envelope can be  written by substituting (\ref{gainalt}) in (\ref{snellsw})  as follows:
\begin{align}
&S_{t}=\esssup_{t' \in \mathfrak{M}_{t}}-\displaystyle\min_{\pi\in\Pi^{\mathsf{online}}}&\nonumber\\
&\mathbb {E}\left[ \displaystyle\int_{a}^{\infty}e^{-\alpha(w-a)}p(\Delta(w))  dw \middle|Z_{k}=a, Z_{k+1}=t', \mathcal{F}_{t}\right]. \nonumber\\
\end{align}
Hence,
\begin{align}
&S_{t}=
-\displaystyle\min_{\pi\in\Pi^{\mathsf{online}}}\nonumber\\
&\mathbb {E}\left[ \int_{a}^{\infty}e^{-\alpha(w-a)}p(\Delta(w))  dw \middle|Z_{k}=a, Z_{k+1}\geq t, \mathcal{F}_{t}\right]. 
\end{align}

Accordingly, using the definition of $J_{\alpha}(\Delta(a),E(a))$, we can write
\begin{equation}
\label{snellcost} 
S_{t}=-\int_{a}^{t}e^{-\alpha(w-a)}p(w-a) dw +e^{-\alpha(t-a)}J_{\alpha}(\Delta(t), E(t)).
\end{equation}

Therefore, as the first terms in  (\ref{gaincost}) and (\ref{snellcost}) are identical, the optimal stopping rule in (\ref{rulesnell}) is equivalent to
\begin{equation}
\label{ruleage}
Z_{k+1} = \inf \lbrace t \geq Z_{k}=a : J_{\alpha}(0,E(t)-1) = J_{\alpha}(\Delta(t),E(t))\rbrace.
\end{equation}
Next, we show that the stopping rule in (\ref{ruleage}) is a threshold rule in age. In order to show this, let us define the function $\rho_{\alpha}(\cdot):\lbrace 0,1,...,B\rbrace \rightarrow [0,\infty)$ such that:
\begin{equation}
\label{rhoalpha}
\rho_{\alpha}(\ell): =\inf\lbrace w \geq 0 : J_{\alpha}(0,\ell-1) = J_{\alpha}(w,\ell)\rbrace.
\end{equation}

We can show that for any $\Delta \geq \rho_{\alpha}(\ell)$, it is guaranteed that $J_{\alpha}(0,\ell-1) = J_{\alpha}(\Delta ,\ell)$ due to the following reasons:
\begin{itemize}
\item For any $\Delta$ and $\ell \in \lbrace 0, 1,2,.., B\rbrace$, $J_{\alpha}(\Delta ,\ell)$ is smaller than or equal to $J_{\alpha}(0,\ell-1)$ as :
\begin{align*}
&J_{\alpha}(\Delta ,\ell)=\displaystyle\min_{\pi\in\Pi^{\mathsf{online}}} e^{a}&\\
&\mathbb {E}\left[ \displaystyle\int_{a}^{\infty}e^{-\alpha w}p(\Delta(w))  dw \middle| Z_{k}=t_{a}-\Delta , Z_{k+1}\geq t_{a}, E(t_{a})=\ell \right]&\\
&\leq\displaystyle\min_{\pi\in\Pi^{\mathsf{online}}}e^{a}
&\\&\mathbb {E}\left[ \displaystyle\int_{a}^{\infty}e^{-\alpha w}p(\Delta(w))  dw \middle| Z_{k}=t_{a}-\Delta, Z_{k+1}=a , E(t)=\ell \right]&\\
&=J_{\alpha}(0,\ell-1),&
\end{align*}
where the inequality is true as the expectation is conditioned on policies with $Z_{k+1}=t_{a}$.
\item For any $\ell \in \lbrace 0, 1,2,.., B\rbrace$, $J_{\alpha}(\Delta ,\ell)$ is non-decreasing in $\Delta$ as : 
\begin{align*}
&J_{\alpha}(\Delta ,\ell) =\displaystyle\min_{\pi\in\Pi^{\mathsf{online}}} &\\
&\mathbb {E}\left[ \displaystyle\int_{a}^{Z_{k+1}}e^{-\alpha(w-a)}p(w+\Delta - t_{a})  dw \middle| \theta(\Delta)\right]&\\
&+\mathbb {E}\left[ \displaystyle\int_{Z_{k+1}}^{\infty}e^{-\alpha(w-a)}p(\Delta(w))  dw \middle| \theta(\Delta)\right]&\\
&\leq\displaystyle\min_{\pi\in\Pi^{\mathsf{online}}}\mathbb {E}\left[ \displaystyle\int_{a}^{Z_{k+1}}e^{-\alpha(w-a)}p(w+\Delta' - t_{a})  dw \middle| \theta(\Delta) \right]&\\
&+\mathbb {E}\left[ \displaystyle\int_{Z_{k+1}}^{\infty}e^{-\alpha(w-a)}p(\Delta(w))  dw \middle|\theta(\Delta\right]&\\
&=\displaystyle\min_{\pi\in\Pi^{\mathsf{online}}}\mathbb {E}\left[ \displaystyle\int_{a}^{Z_{k+1}}e^{-\alpha(w-a)}p(\Delta(w))  dw \middle| \theta(\Delta') \right]&\\
&+\mathbb {E}\left[ \displaystyle\int_{Z_{k+1}}^{\infty}e^{-\alpha(w-a)}p(\Delta(w))  dw \middle|\theta(\Delta') \right]&\\
&=J_{\alpha}(\Delta',\ell),&
\end{align*}
for any $\Delta' \geq \Delta$ and $\theta(\Delta) :=(Z_{k}=t_{a}-\Delta , Z_{k+1}\geq t_{a}, E(t_{a})=\ell)$ where the inequality follows from the fact that $p(\cdot)$ is non-decreasing and the second equality is due to that,  given $Z_{k+1}$, the integrated values are conditionally independent from  $Z_{k}$.
\end{itemize}

Accordingly, $J_{\alpha}(\Delta,\ell)=J_{\alpha}(0,\ell-1)$ for any $\ell \in \lbrace 0, 1,2,.., B\rbrace$ and $\Delta \geq \rho_{\alpha}(\ell)$. Therefore, the stopping rule in (\ref{ruleage}) is equivalent to:
\begin{equation}
\label{ruleaged}
Z_{k+1} = \inf \lbrace t \geq Z_{k}= a : \Delta(t)\geq \rho_{\alpha}(E(t))\rbrace,
\end{equation}
for $\ell \in \lbrace 0, 1,2,.., B\rbrace$.

We showed that the stopping rule in (\ref{ruleaged}) gives the optimal stopping time $Z_{k+1}$  for a policy solving  (\ref{kfinite}). Now, we can start discussing \emph{Step 2} in order to show that the optimal stopping rule with the same structure also gives a solution to (\ref{minavaged}). 

In this part (\emph{Step 2}) of the proof, we will consider the optimal stopping rules  in (\ref{ruleaged}) while the discount factor $\alpha$ is vanishing to zero. Notice that the policy solving (\ref{kfinite}) is identified by $\rho_{\alpha}(\ell)$ due to (\ref{ruleaged}). Let $\pi_{\alpha}$ and $\Delta_{\pi_{\alpha}}(t)$ be a policy obeying (\ref{ruleaged}) and solving (\ref{kfinite}) for discount factor $\alpha$ and the age at time $t$ for that policy, respectively. We will show the following

\begin{eqnarray}
\label{pavbeta}
&\lim_{\beta\downarrow0}\lim_{t_f \rightarrow \infty}\frac{\int_{0}^{t_f}\mathbb {E}\left[p(\Delta_{\pi_{\beta}}(t))\right]dt}{t_f} =&  \nonumber\\
&\inf_{\pi\in\Pi^{\mathsf{online}}}\limsup_{t_f \rightarrow \infty}\frac{\int_{0}^{t_f}\mathbb {E}\left[p(\Delta_{\pi}(t))\right]dt}{t_f},&
\end{eqnarray}
which implies that for any $\lbrace \beta_{n}\rbrace_{n\geq 1}\downarrow0$ sequence, $\pi_{\beta_{n}}$ converges to the policy solving (\ref{minavaged}).

To prove the equivalence in (\ref{pavbeta}), we will use Feller's Tauberian theorem \cite{feller1971} (also see the Tauberian theorem in \cite{widder1946}) which can be stated as follows:
\begin{theorem}
\textbf{(Feller 1971)}
\label{Fellerstauberiantheorem} 
Let $f(t)$ be a Lebesgue-measurable, bounded, real function. Then,
\begin{align}
\label{tauberianFeller}
&\liminf_{t_f \rightarrow \infty}\frac{\int_{0}^{t_f}f(t)dt}{t_f}\leq \liminf_{\alpha\downarrow0}\alpha \int_{0}^{t_f}e^{-\alpha (t-a)}f(t)dt\ \nonumber\\
&\leq \limsup_{\alpha\downarrow0}\alpha \int_{0}^{t_f}e^{-\alpha (t-a)}f(t)dt \leq  \limsup_{t_f \rightarrow \infty}\frac{\int_{0}^{t_f}f(t)dt}{t_f}. \end{align}
Moreover, if the central inequality is an equality, then all inequalities are equalities.
\end{theorem}

This theorem can be applied for the function $f(t)=\mathbb {E}\left[p(\Delta_{\pi_{\beta}}(t))\right]$ where $\beta>0$ \footnote{Note that the function $\mathbb {E}\left[p(\Delta_{\pi_{\beta}}(t))\right]$ is Lebesgue-measurable (as $p(\cdot)$ is non-decreasing) and bounded (as $X_{k}$s are bounded w.p.1 for a policy obeying (\ref{ruleaged})).}. To simplify the inequalities for this case, let's define a function $J_{\alpha;\beta}(\Delta(a),E(a))$ for $\beta > 0$ such that:
\begin{align}
&J_{\alpha;\beta}(\Delta(a),E(a)) : = \nonumber\\
&\int_{a}^{\infty} e^{-\alpha(t-a)}\mathbb {E}\left[p(\Delta_{\pi_{\beta}}(t))\middle| \Delta(a),E(a)\right]dt.
\end{align}
Note that for $a=0$:
\[
J_{\alpha;\beta}(0,0) : = \int_{a}^{\infty} e^{-\alpha (t-a)}\mathbb {E}\left[p(\Delta_{\pi_{\beta}}(t))\right]dt.
\]
Accordingly, we can apply Feller's Tauberian theorem for $f(t)=\mathbb {E}\left[p(\Delta_{\pi_{\beta}}(t))\right]$ when $a=0$ giving:
\begin{eqnarray}
\label{tauberianJ}
 &\liminf_{t_f \rightarrow \infty}\frac{\int_{0}^{t_f}\mathbb {E}\left[p(\Delta_{\pi_{\beta}}(t))\right]dt}{t_f}\leq \liminf_{\alpha\downarrow0}\alpha J_{\alpha;\beta}(0,0)\leq & \nonumber\\&\limsup_{\alpha\downarrow0}\alpha J_{\alpha;\beta}(0,0)
 \leq  \limsup_{t_f \rightarrow \infty}\frac{\int_{0}^{t_f}\mathbb {E}\left[p(\Delta_{\pi_{\beta}}(t))\right]dt}{t_f}.&\nonumber\\
 \end{eqnarray}

We can show that the inequalities in (\ref{tauberianJ}) are satisfied with equality for any $\pi_{\beta}$ with $\beta > 0$ as $ \lim_{t_f \rightarrow \infty}\frac{\int_{0}^{t_f}\mathbb {E}\left[\Delta_{\pi_{\beta}}(t)\right]dt}{t_f}$ exists for any $\pi_{\beta}$ with $\beta > 0$. To see this, consider the following lemma:
\begin{lemma}
\label{averageageforalphalemma}
For $\alpha>0$ and $\lbrace Z_{k+1},k \geq 0\rbrace$ with $Z_{k+1}$ as in (\ref{ruleaged}), the following holds:
\begin{align}
\label{averageageforalpha}
&\lim_{t_f \rightarrow \infty}\frac{\int_{0}^{t_f}\mathbb {E}\left[p(\Delta_{\pi_{\alpha}}(t))\right]dt}{t_f}= \nonumber\\
&\frac{\lim_{n\rightarrow +\infty}\frac{1}{n}\sum_{k=0}^{n}\mathbb {E}[p(X_{k})]}{\lim_{n\rightarrow +\infty}\frac{1}{n}\sum_{k=0}^{n}X_{k}}\mbox{\:\: w.p.1.}
\end{align} 
\end{lemma}
\begin{proof}
The proof of Lemma \ref{stoppingexists} showed that for $Z_{k}=a$ and optimal stopping time solving (\ref{stoppingae}) it is true that $\Pr(X_{k+1}\geq x)\leq \Pr(t_{c}-t_{a}+Y_{B}\geq x)$ where $t_{c}$ is the deterministic time defined in (\ref{tcdef}) and $Y_{B}$ is an Erlang distributed with parameter $B$ which obeys (\ref{erlangdist}). Accordingly, $\mathbb {E}[p(X_{k+1})]$ is finite as $\mathbb {E}[p(\alpha J_{\alpha}(0,B)+Y_{B})]$ is finite for $\alpha>0$. On the other hand, $\lim_{n\rightarrow +\infty}\frac{1}{n}\sum_{k=0}^{n}X_{k}< \infty$ w.p.1 and  $\lim_{n\rightarrow +\infty}\frac{1}{n}\sum_{k=0}^{n}X_{k}>\frac{1}{\mu_{H}}$ w.p.1 due to the energy causality constraint. Therefore, we can apply the derivation steps in \cite[Theorem 5.4.5]{gallager2013stochastic} and obtain (\ref{averageageforalpha}). This completes the proof.
\end{proof}
Lemma \ref{averageageforalphalemma} and (\ref{tauberianJ}) imply the following for for $a=0$ and $\beta > 0$:
\begin{equation}
\label{avageforalphabeta}
\lim_{\alpha\downarrow0}\alpha J_{\alpha;\beta}(0,0)=  \lim_{t_f \rightarrow \infty}\frac{\int_{0}^{t_f}\mathbb {E}\left[p(\Delta_{\pi_{\beta}}(t))\right]dt}{t_f}.
\end{equation}

Now, consider an arbitrary online policy $\pi$ for which $\mathbb {E}\left[p(\Delta_{\pi}(t))\right]$ is Lebesgue-measurable and bounded, then apply Feller's Tauberian theorem for $f(t)=\mathbb {E}\left[p(\Delta_{\pi}(t))\right]$ giving the following inequality when $t_{a}=0$:
\begin{eqnarray}
\label{pavforarbitraryonline}
&\limsup_{\alpha\downarrow0}\alpha\int_{0}^{\infty} e^{-\alpha (t-a)}\mathbb {E}\left[p(\Delta_{\pi}(t))\right]dt \leq & \nonumber\\ &\limsup_{t_f \rightarrow \infty}\frac{\int_{0}^{t_f}\mathbb {E}\left[p(\Delta_{\pi}(t))\right]dt}{t_f}.&
\end{eqnarray} 

Note that for $\alpha > 0$, $J_{\alpha;\beta}(0,0)$ is minimized for $\alpha=\beta$, hence:
\begin{eqnarray}
\label{pavforarbitraryonlinealpha}
&\lim_{\beta\downarrow0}\lim_{\alpha\downarrow0}\alpha J_{\alpha;\beta}(0,0)=\inf_{\beta > 0}\lim_{\alpha\downarrow0}\alpha J_{\alpha;\beta}(0,0)\leq & \nonumber\\
&\limsup_{\alpha\downarrow0}\alpha\int_{0}^{\infty} e^{-\alpha (t-a)}\mathbb {E}\left[p(\Delta_{\pi}(t))\right]dt.&
\end{eqnarray} 
Combining (\ref{avageforalphabeta}), (\ref{pavforarbitraryonline}) and (\ref{pavforarbitraryonlinealpha}), we get (\ref{pavbeta}). This completes the proof.

\subsection{The Proof of Theorem \ref{existopmothreshold}}
\label{proof:existopmothreshold}
Theorem \ref{existopmothreshold} follows from the proof of Theorem \ref{existopthreshold}. To prove the theorem it is sufficient to show  that for any $\alpha > 0$, $\rho_{\alpha}(\ell)$ (see (\ref{rhoalpha})) is non-increasing in $\ell$ as this guarantees that the monotonicity of optimal thresholds holds for any sequence of $\alpha$ values that vanishes to zero. To see this, consider the following lemma and the argument provided below its proof:
\begin{lemma}
\label{submodularitylemmaforJmax}
For $J(\cdot,\cdot)$ is the function defined in (\ref{Jdefinition}), $J_{\alpha}(0,\ell)-J_{\alpha}(0,\ell+1)$ is non-increasing in $\ell \in \lbrace 0,1,...,B-1\rbrace$ for any $\alpha\geq 0$.
\end{lemma}
\begin{proof}

First, consider the alternative formulation of $J_{\alpha}(r,\ell+1)$ in below:
\begin{align*}
&J_{\alpha}(r,\ell+1)= \displaystyle\min_{\pi\in\Pi^{\mathsf{online}}}e^{a}\mathbb {E}[\nonumber\\
&\mathbb {E}\left[ \displaystyle\int_{a}^{\infty}e^{-\alpha t }p(\Delta(t))  dt \middle| Z_{k+1},\Delta(t_{a})=r,E(t_{a})=\ell+1 \right]],
\end{align*}
where the outer expectation is taken over $Z_{k+1}$. 

Let 
\[
K_{r,\ell+1}(z,\sigma):= 
\]
\[
\Pr\left( Z_{k+1}=z, H(z)-H(a)=\sigma \middle| \Delta(t_{a})=r,E(t_{a})=\ell+1 \right)
\]
be the joint distribution of $Z_{k+1} \in \mathfrak{M}_{a}$ and the energy harvested during $[a,z]$. Then, we can write $J_{\alpha}(r,\ell+1)$ as follows:
\begin{align}
\label{Jhalforprev}
&J_{\alpha}(r,\ell+1)=\min_{Z_{k+1} \in \mathfrak{M}_{a}} \sum_{\sigma=0}^{\infty}\int_{t_{a}}^{\infty} K_{r,\ell+1}(z,\sigma)e^{a}\times& \nonumber\\ 
& \left[ \displaystyle\int_{a}^{z}e^{-\alpha t} p(\Delta(t))  dt +e^{-\alpha z }J_{\alpha}(0,\min\lbrace\ell+\sigma,B-1\rbrace)\right]dz.&
\end{align}

Similarly,
\begin{align}
\label{decreasedJhal}
&J_{\alpha}(r,\ell+2)=\min_{Z_{k+1} \in \mathfrak{M}_{a}} \sum_{\sigma=0}^{\infty}\int_{t_{a}}^{\infty} K_{r,\ell+2}(z,\sigma)e^{a}\times& \nonumber\\ 
& \left[ \displaystyle\int_{a}^{z}e^{-\alpha t}p(\Delta(t))  dt \!\!+\!\! e^{-\alpha z}J_{\alpha}(0,\min\lbrace\ell\!\!+\!\! 1 \!\!+\!\! \sigma,B-1\rbrace)\right]dz.
\end{align}
Now, let $K_{r,\ell+2}^{*}(z,\sigma)$ be the distribution corresponding to the update time $Z_{k+1} \in \mathfrak{M}_{a}$ that is optimal in (\ref{decreasedJhal}), which means:
\begin{align}
\label{decreasedJhalop}
&J_{\alpha}(r,\ell+2)= \sum_{\sigma=0}^{\infty}\int_{t_{a}}^{\infty} K_{r,\ell+2}^{*}(z,\sigma)e^{a}\times& \nonumber\\ 
& \left[  \displaystyle\int_{a}^{z}e^{-\alpha t }p(\Delta(t))  dt+e^{-\alpha z }J_{\alpha}(0,\min\lbrace\ell\!\!+\!\! 1 \!\!+\!\! \sigma,B-1\rbrace)\right]dz.&
\end{align}
Clearly, $K_{r,\ell+2}^{*}(z,\sigma)$ is not necessarily the joint distribution corresponding the update time $Z_{k+1} \in \mathfrak{M}_{a}$ that is optimal for (\ref{Jhalforprev}), hence: 
\begin{align}
\label{Jhalop}
&J_{\alpha}(r,\ell+1)\leq\sum_{\sigma=0}^{\infty}\int_{t_{a}}^{\infty} K_{r,\ell+2}^{*}(z,\sigma)[ \displaystyle\int_{a}^{z}e^{-\alpha(t-a)}p(\Delta(t))  dt  \nonumber\\ 
& +e^{-\alpha(z-a)}J_{\alpha}(0,\min\lbrace\ell+\sigma,B-1\rbrace)]dz.
\end{align}
Combining (\ref{decreasedJhalop}) and (\ref{Jhalop}) gives:
\begin{align}
\label{decreasedJhalop}
&J_{\alpha}(r,\ell+1)- J_{\alpha}(r,\ell+2)\leq \sum_{\sigma=0}^{\infty}\int_{a}^{\infty}K_{r,\ell+2}^{*}(z,\sigma)e^{-\alpha(z-a)}\times& \nonumber\\ 
&[J_{\alpha}(0,\min\lbrace\ell+\sigma,\! B-1\rbrace)\!\!-\!J_{\alpha}(0,\min\lbrace\ell+1+\sigma,\! B-1\rbrace)]dz.&
\end{align}
which implies :
\begin{eqnarray*}
\label{incrementboundformaxJ}
&J_{\alpha}(r,\ell+1)- J_{\alpha}(r,\ell+2)\leq \max_{\sigma \in \lbrace 0,1,..,B-\ell\rbrace}& \nonumber\\&J_{\alpha}(0,\min\lbrace\ell+\sigma,\! B-1\rbrace)\!-\!J_{\alpha}(0,\min\lbrace\ell\!\!+\!\! 1 \!\!+\!\! \sigma,\! B-1\rbrace)&
\end{eqnarray*}
Now, consider the case when $r=0$ and $\ell=B-2$ for (\ref{decreasedJhalop}):
\begin{align}
&J_{\alpha}(0,B-1)- J_{\alpha}(0,B)\leq  \nonumber\\ 
&\sum_{\sigma=0}^{\infty}\int_{a}^{\infty}K_{r,\ell+2}^{*}(z,\sigma)e^{-\alpha(z-a)}[J_{\alpha}(0,\min\lbrace  B- \!2 \!+ \!\sigma,\! B-1\rbrace)-\nonumber\\
&\!J_{\alpha}(0,\min\lbrace B \!- \! 1 \!+\sigma,\! B-1\rbrace)]dz,
\end{align}
which implies:
\begin{equation}
J_{\alpha}(0,B-1)- J_{\alpha}(0,B)\leq J_{\alpha}(0,B-2)- J_{\alpha}(0,B-1).
\end{equation}
Suppose that the inequality below is true for $j \geq \ell+1$:
\begin{equation}
\label{inductionhypothesirenergyincrements}
J_{\alpha}(0,j+1)- J_{\alpha}(0,j+2)\leq J_{\alpha}(0,j)- J_{\alpha}(0,j+1).
\end{equation}
Then, we have:
\begin{align}
&J_{\alpha}(0,\ell+1)- J_{\alpha}(0,\ell+2)\leq\nonumber\\ &\leq \sum_{\sigma=0}^{\infty}\int_{a}^{\infty}K_{r,\ell+2}^{*}(z,\sigma)
e^{-\alpha(z-a)}\times \nonumber\\ & [J_{\alpha}(0,\min\lbrace\ell+\sigma,\! B-1\rbrace)\!\!-\!J_{\alpha}(0,\min\lbrace\ell+1+\sigma,\! B-1\rbrace)]dz \nonumber\\ 
&\leq \int_{a}^{\infty}K^{*}(z,0)
e^{-\alpha(z-a)}[J_{\alpha}(0,\ell)- J_{\alpha}(0,\ell+1)]dz +\sum_{\sigma=1}^{\infty}\nonumber\\ 
&\int_{a}^{\infty}K_{r,\ell+2}^{*}(z,\sigma)
e^{-\alpha(z-a)}[J_{\alpha}(0,\ell+1)- J_{\alpha}(0,\ell+2)]dz \nonumber\\ 
&\leq J_{\alpha}(0,\ell)- J_{\alpha}(0,\ell+1).
\end{align}
This means that the inequality (\ref{inductionhypothesirenergyincrements}) is also true for $j=\ell$ so is for any $j=0,1,...,B-2$ by induction. Combining this and (\ref{incrementboundformaxJ}):
\begin{equation}
\label{ineqJd}
J_{\alpha}(r,\ell+1)- J_{\alpha}(r,\ell+2)\leq J_{\alpha}(0,\ell)- J_{\alpha}(0,\ell+1),
\end{equation}
for $\alpha\geq 0$, $r \geq 0$ and 
\end{proof}
Lemma \ref{submodularitylemmaforJmax} shows that $\rho_{\alpha}(\ell)$ is non-increasing in $\ell$ for $\alpha > 0$. It is sufficient to consider (\ref{ineqJd}) when $r=\rho_{\alpha}(\ell)$:
\begin{equation}
0 = J_{\alpha}(0,\ell-1)-J_{\alpha}(\rho_{\alpha}(\ell),\ell)\leq J_{\alpha}(0,\ell-2)-J_{\alpha}(\rho_{\alpha}(\ell),\ell-1),
\end{equation}
which implies $\rho_{\alpha}(\ell-1)\geq \rho_{\alpha}(\ell)$ combining
\[
J_{\alpha}(0,\ell-2)-J_{\alpha}(\rho_{\alpha}(\ell-1),\ell-1)
\]
and that $J_{\alpha}(r,\ell-1)$ is non-decreasing \footnote{This fact is provided in the proof of Theorem \ref{existopthreshold}.} in $r$.
Accordingly, the optimal policies solving (\ref{adiscounted}) are monotone threshold policies, i.e., $\pi_{\alpha} \in \Pi^{MT}$ for any $\alpha > 0$.

\subsection{The proof of Lemma \ref{diffmoments}}
\label{proof:diffmoments}

Let $\tau_{B+1}=0$. Then, consider:

\begin{align*}
&\frac{\partial}{\partial \tau_{i}}\mathbb {E}\left[ X^{2}\right]=\frac{\partial}{\partial \tau_{i}} \int_{0}^{\infty}2x\Pr(X \geq x)dx\\
&=\frac{\partial}{\partial \tau_{i}}\displaystyle\sum_{i=0}^{B}\int_{\tau_{i+1}}^{\tau_{i}}2x\Pr(X\geq x)dx\\
&=2\frac{\partial}{\partial \tau_{i}}\left[\int_{\tau_{i+1}}^{\tau_{i}}x\Pr(X\geq x)dx
+\int_{\tau_{i}}^{\tau_{i-1}}x\Pr(X\geq x)dx\right]\\
&=2\tau_{i}\frac{\partial}{\partial \tau_{i}}\int_{\tau_{i+1}}^{\tau_{i-1}}\Pr(X\geq x)dx\\
&=2\tau_{i}\frac{\partial}{\partial \tau_{i}}\displaystyle\sum_{i=0}^{B}\int_{\tau_{i+1}}^{\tau_{i}}\Pr(X\geq x)dx=2\tau_{i}\frac{\partial}{\partial \tau_{i}}\mathbb {E}\left[ X\right],\\
\end{align*}
for $i=0,1,...,B$.

\subsection{Useful Results for Asymptotic Properties}
Lemma \ref{ergodicitylemma}, \ref{ergoupdates} and \ref{convage} provide some useful results that combine ergodicity properties and renewal-reward theorem for a DTMC with transition probabilities in (\ref{transitionerlang}). 

\begin{lemma}
The DTMC with the transition probabilities in (\ref{transitionerlang}) is ergodic for a monotone threshold policy where  $\tau_{1}$ is finite.
\label{ergodicitylemma}
\end{lemma}
\begin{proof}
Consider an energy state $j$ in $[0,B-1]$. We will show that any other energy state $i$ is reachable from $j$ in at most $B-1$ steps with a positive probability. For $i\geq j$, the higher energy state $i$ is reachable from $j$ in one step with a positive probability as for $i=B-1$, $\Pr(Y_{B-j}\leq \tau_{B-1})$ is strictly positive and for $ j\leq i<B-1$:
\begin{align*}
&\Pr(Y_{1+i-j}\leq \tau_{i})- \Pr(Y_{2+i-j}\leq \tau_{i+1})\geq\\
&\Pr(Y_{1+i-j}\leq \tau_{i+1})- \Pr(Y_{2+i-j}\leq \tau_{i+1})>0,
\end{align*}

as $\tau_{i+1}\leq \tau_{i}$ and $i-j \geq 0$.

Similarly, the energy state $i=j-1$ for $j=1,...., B-1$ can be reached from $j$ with a probability $1-\Pr(Y_{1}\leq \tau_{j})$ which is strictly positive as $\tau_{j}$ is finite. This means that any state $i<j$ can be reached from $j$ in at most $B-1$ steps with a positive probability. 
\end{proof}
\begin{lemma}
\label{ergoupdates}
For monotone threshold policies with finite $\tau_{1}$,  the following is true:
\begin{equation}
\lim_{n\rightarrow +\infty}\frac{1}{n}\sum_{k=0}^{n}X_{k}=\sum_{j=0}^{B-1}\mathbb {E}\left[ X \mid E=j\right] \Pr(E=j) \mbox{\:\: w.p.1.}
\label{monotoneX}
\end{equation}
\begin{equation}
\lim_{n\rightarrow +\infty}\frac{1}{2n}\sum_{k=0}^{n}\mathbb {E}[X_{k}^{2}]=\frac{1}{2}\sum_{j=0}^{B-1} \mathbb {E}\left[ X^{2} \mid E=j\right] \Pr(E=j),
\label{monotoneC}
\end{equation}
where  $\Pr(E=j)$ is the steady-state probability for energy state $j$, $\mathbb {E}\left[ X \mid E=j\right] \triangleq \mathbb {E}\left[ X_{k} \mid E(Z_{k})=j\right]$ and $\mathbb {E}\left[ X^{2} \mid E=j\right] \triangleq \mathbb {E}\left[ X_{k}^{2} \mid E(Z_{k})=j\right]$.
\end{lemma}
\begin{proof}
Consider:
\[
\frac{1}{n}\sum_{k=0}^{n} X_{k}=\frac{1}{n}\sum_{j=0}^{B-1}\sum_{\substack{k\in [0,n]\\ E(Z_{k})=j }} X_{k}=\frac{1}{n}\sum_{j=0}^{B-1}\displaystyle\sum_{\ell=0}^{L_{j}} X_{\ell;j},
\]
where $L_{j}$ is the number of $k$s in $[0,n]$ such that $E(Z_{k})=j$ and $X_{\ell;j}$ is a r.v. with the CDF $\Pr(X_{\ell;j}\leq x)= \Pr(X_{k}\leq x \mid  E(Z_{k})=j)$ for some $k$.

Note that the sequence $X_{0;j}, X_{1;j}, ...,X_{L_{j};j}$ is i.i.d. for any $j$ and their mean is bounded as all thresholds are finite, hence:
\[
\displaystyle\lim_{L_{j}\rightarrow \infty} \frac{1}{L_{j}}\displaystyle\sum_{\ell=0}^{L_{j}} X_{\ell;j}= \mathbb {E}\left[ X \mid E=j\right], w.p.1.
\]
Due to the ergodicity of $E(Z_{k})$s (Lemma \ref{ergodicitylemma}):
\[
\displaystyle\lim_{n\rightarrow \infty}\frac{L_{j}}{n}=\Pr(E=j), w.p.1.
\]
Therefore,
\begin{align*}
\displaystyle\lim_{n\rightarrow \infty}\frac{1}{n}\displaystyle\sum_{k=0}^{n} X_{k}&=\lim_{n\rightarrow \infty}\displaystyle\sum_{j=0}^{B-1}\frac{L_{j}}{n}(\frac{1}{L_{j}}\displaystyle\sum_{\ell=0}^{L_{j}} X_{\ell;j}),\\
&=\displaystyle\sum_{j=0}^{B-1} \mathbb {E}\left[ X \mid E=j\right] \Pr(E=j), w.p.1.
\end{align*}
Similarly,
\begin{align*}
\displaystyle\lim_{n\rightarrow \infty}\frac{1}{n}\displaystyle\sum_{k=0}^{n} \mathbb {E}[X_{k}^{2}]&=\lim_{n\rightarrow \infty}\displaystyle\sum_{j=0}^{B-1}\frac{L_{j}}{n}(\frac{1}{L_{j}}\displaystyle\sum_{\ell=0}^{L_{j}} X_{\ell;j}^{2})\\
&=\displaystyle\sum_{j=0}^{B-1} \mathbb {E}\left[ X^{2} \mid E=j\right] \Pr(E=j), w.p.1.
\end{align*}
\end{proof}
\begin{lemma}
\label{convage}
For a threshold policy where $\tau_{1}$ is finite, the average age $\bar{\Delta}$ is finite (w.p.1) and given by the following expression.
\begin{equation}
\bar{\Delta}=\frac{\lim_{n\rightarrow +\infty}\frac{1}{2n}\sum_{k=0}^{n}\mathbb {E}[X_{k}^{2}]}{\lim_{n\rightarrow +\infty}\frac{1}{n}\sum_{k=0}^{n}X_{k}} \mbox{\:\: w.p.1.}
\end{equation}
\end{lemma}
\begin{proof}
The proof is a generalization of Theorem 5.4.5 in \cite{gallager2013stochastic} for the case where $X_{k}$s are non-i.i.d. but the limits still exist (w.p.1). When $X_{k}$s  are i.i.d. with $\mathbb {E}[X_{k}]< \infty$ and $\mathbb {E}[X_{k}^{2}]< \infty$, the convergence (w.p.1) of the limits is guaranteed.
\end{proof}
\subsection{The proof of Theorem \ref{fixedpthreshold}}
\label{proof:fixedpthreshold}
Theorem \ref{fixedpthreshold} follows from the proof of Theorem \ref{existopthreshold}. The proof of Lemma \ref{stoppingexists} shows that given that  $Z_{k}=a$ is  the last update time and $E(t')=B$ for some $t'>a$, the condition $S_{t}=G_{t}$ is satisfied for the first time when $t\geq \lbrace t', t_{c}\rbrace$ (see (\ref{tcdef})). This means that 
$\rho_{\alpha}(B)=\alpha J_{\alpha}(0,B-1)$ for $\rho_{\alpha}(E(t))$ in (\ref{ruleaged}). 
Accordingly,
\begin{eqnarray*}
&p(\tau_{B}^{*})=\lim_{\alpha\downarrow 0}\rho_{\alpha}(B)=\lim_{\alpha\downarrow 0}\alpha J_{\alpha}(0,B-1)=&\nonumber\\&\min_{\pi\in\Pi^{\mathsf{online}}}\limsup_{t_f \rightarrow \infty}\frac{\int_{0}^{t_f}\mathbb {E}\left[p(\Delta_{\pi}(t))\mid E(0)=B\right]dt}{t_f}=\bar{p}_{\pi^*},&
\end{eqnarray*}
which follows from the application of Feller's Tauberian theorem (applying Theorem \ref{Fellerstauberiantheorem} for $f(t)=\mathbb {E}\left[p(\Delta_{\pi}(t))\mid E(0)=B\right]$). This completes the proof. 
\subsection{The Proof of Theorem \ref{B1age}}
\label{proof:B1age}
By Lemma \ref{convage} and Lemma \ref{ergoupdates}, $\bar{\Delta}$ for $B=1$ can be computed as follows
\begin{equation}
\label{B1delta}
\bar{\Delta}=\frac{1}{2}\frac{\mathbb {E}\left[ X^{2} \mid E=0\right]\Pr(E=0)}{\mathbb {E}\left[ X \mid E=0\right]\Pr(E=0)},
\end{equation}
where $\Pr(E=0)=1$, $\mathbb {E}\left[ X^{2} \mid E=0\right]=\tau_{1}^{2}+(\frac{2}{\mu_{H}^{2}}+\frac{2}{\mu_{H}}\tau_{1})e^{-\mu_{H}\tau_{1}}$ and $\mathbb {E}\left[ X \mid E=0\right]=\tau_{1}+$ $\frac{1}{\mu_{H}}e^{-\mu_{H}\tau_{1}}$. Accordingly, $\bar{\Delta}$ is given by (\ref{B1avage}).
%\begin{equation}
%\bar{\Delta}=\frac{\frac{1}{2}(\mu_{H}\tau_{1})^{2}+ e^{-\mu_{H}\tau_{1}}\!(\mu_{H}\tau_{1}+1)}{\mu_{H}(\mu_{H}\tau_{1}+e^{-\mu_{H}\tau_{1}})}.
%\end{equation}
By Theorem \ref{fixedpthreshold}, $\tau_{1}^{*}=\bar{\Delta}_{\pi^*}$ and combining this with (\ref{B1avage}) results in
\begin{equation}
\label{B1optimal}
\mu_{H}\tau_{1}^{*}=\frac{\frac{1}{2}(\mu_{H}\tau_{1}^{*})^{2}+ e^{-\mu_{H}\tau_{1}^{*}}\!(\mu_{H}\tau_{1}+1)}{\mu_{H}\tau_{1}^{*}+e^{-\mu_{H}\tau_{1}^{*}}}.
\end{equation}
Solving (\ref{B1optimal}) gives that $(\tau_{1}^{*})^{2}=\frac{2}{\mu_{H}}e^{-\mu_{H}\tau_{1}^{*}}$ which means $\tau_{1}^{*}=\frac{1}{\mu_{H}}2W(\frac{1}{\sqrt{2}})$.
\subsection{The Proof of Theorem \ref{B2age}}
\label{proof:B2age}
%See \cite{AoIISITtechnicalreport}.

By Lemma \ref{convage} and Lemma \ref{ergoupdates}, $\bar{\Delta}$ for $B=2$ is the following:
\begin{align}
\label{B2delta}
&\bar{\Delta}=\nonumber\\
&\frac{1}{2}\frac{\mathbb {E}\left[ X^{2} \mid E=0\right]\Pr(E=0)+\mathbb {E}\left[ X^{2} \mid E=1\right]\Pr(E=1)}{\mathbb {E}\left[ X \mid E=0\right]\Pr(E=0)+\mathbb {E}\left[ X \mid E=1\right]\Pr(E=1)}.
\end{align}
The probability of being in $E=1$, i.e. $\Pr(E=1)$ can be solved using:
\begin{equation}
\label{steadyE}
\Pr(E=1)=\sum_{j=0}^{1}\Pr(E(Z_{k+1})=1 \mid  E(Z_{k})=j)\Pr(E=j).
\end{equation}
Combining (\ref{steadyE}) and (\ref{cdferlang}),
\begin{equation}
\label{steadyexp}
\Pr(E=1)=\frac{e^{-\mu_{H}\tau_{1}}}{1-e^{-\mu_{H}\tau_{1}}\mu_{H}\tau_{1}}.
\end{equation}
Now, we can obtain $\mathbb {E}\left[ X^{2} \mid E=j\right]$, $\mathbb {E}\left[ X \mid E=j\right]$ using (\ref{cdferlang}). Combining these with (\ref{steadyexp}) and substituting in (\ref{B2delta}) gives  (\ref{B2age}).

%By Lemma \ref{convage} and Lemma \ref{ergoupdates}, the average age for $B=2$ is the following:
%\begin{equation}
%\label{B2delta}
%\bar{\Delta}=\frac{1}{2}\frac{\mathbb {E}\left[ X^{2} \mid j=0\right]\Pr(E=0)+\mathbb {E}\left[ X^{2} \mid j=1\right]\Pr(E=1)}{\mathbb {E}\left[ X \mid j=0\right]\Pr(E=0)+\mathbb {E}\left[ X \mid j=1\right]\Pr(E=1)}.
%\end{equation}
%The probability of being in $E=1$, i.e. $\Pr(E=1)$ can be solved using:
%\begin{equation}
%\label{steadyE}
%\Pr(E=1)=\sum_{j=0}^{1}\Pr(E(Z_{k+1})=1 \mid  E(Z_{k})=j)\Pr(E=j).
%\end{equation}
%Combining (\ref{steadyE}) and (\ref{cdferlang}),
%\[
%\Pr(E=1)=\frac{\Pr(Y_{2}\leq \tau_{1})+\Pr(Y_{1}\leq \tau_{1})}{1+\Pr(Y_{2}\leq \tau_{1})-\Pr(Y_{1}\leq \tau_{1})}.
%\]
%Accordingly,
%\begin{equation}
%\label{steadyexp}
%\Pr(E=1)=\frac{e^{-\mu_{H}\tau_{1}}}{1-e^{-\mu_{H}\tau_{1}}\mu_{H}\tau_{1}}.
%\end{equation}
%Now, we can obtain $\mathbb {E}\left[ X^{2} \mid j\right]$, $\mathbb {E}\left[ X \mid j\right]$ using (\ref{cdferlang}). Combining these with (\ref{steadyexp}) and substituting in (\ref{B2delta}) gives  (\ref{B2age}).
\subsection{The Proof of Theorem \ref{momentsthreshold}}
First, we show that $\tau_{B}\geq \bar{\Delta}_{\pi^{*}}$ is \emph{necessary} to find a solution to (\ref{transcendentalage}) with  monotonic non-increasing thresholds. Then, we show that this condition is also \emph{sufficient}.

The \emph{necessity} part of the proof follows from the fact that $\tau_{B}=\bar{\Delta}_{\pi}$ for any solution of (\ref{transcendentalage}),  as $\bar{\Delta}_{\pi}=m_{1}(\tau_{1},\tau_{2},...,\tau_{B})/2m_{2}(\tau_{1},\tau_{2},...,\tau_{B})$ by  Lemma \ref{convage} and Lemma \ref{ergoupdates}. Therefore, by the optimality of $ \bar{\Delta}_{\pi^{*}}$, $\tau_{B}\geq \bar{\Delta}_{\pi^{*}}$ must hold for any solution of (\ref{transcendentalage}).

Now, we consider the \emph{sufficiency} part of the proof where it is useful to define a function $\phi :[0,\infty)^{B}\rightarrow \mathbb{R}$ as follows:
\begin{eqnarray*}
&\phi(\tau_{B}, \tau_{B-1}-\tau_{B},..., \tau_{1}-\tau_{2})\triangleq\\ &2\tau_{B}m_{1}(\tau_{1},\tau_{2},...,\tau_{B})-m_{2}(\tau_{1},\tau_{2},...,\tau_{B}).  
\end{eqnarray*}
Using this definition, (\ref{transcendentalage}) can be written as,
\begin{equation*}
    \phi(\tau_{B}, \tau_{B-1}-\tau_{B},..., \tau_{1}-\tau_{2})=0.
\end{equation*}

We need to show that given $\tau_{B}\geq \bar{\Delta}_{\pi^{*}}$, one can find a set of non-negative real numbers $d_{1},....,d_{B-1}$ such that $\phi(\tau_{B}, d_{B-1},..., d_{1})=0$. Accordingly, $\tau_{B}$ and $d_{1},....,d_{B-1}$  constitute a solution to (\ref{transcendentalage}) with monotonic non-decreasing thresholds where $\tau_{i}=\tau_{i+1}+d_{i}$, for  $i=1,..., B-1$. In order to prove this, let us start with the optimal policy $\pi^{*}=(\tau_{1}^{*}, \tau_{2}^{*}...,\tau_{B}^{*})$ where we know that $\tau_{B}^{*}=\bar{\Delta}_{\pi^{*}}$ by Theorem \ref{fixedpthreshold}. Starting from the optimal policy $\pi^{*}$, the policy will be modified following the procedure below:
\begin{itemize}
\item \emph{Phase 1:} Modify the policy $\pi^{(+)}=(\tau_{1}^{(+)}, \tau_{2}^{(+)}...,\tau_{B}^{(+)})$ from the previous phase to the policy $\pi^{(-)}=(\tau_{1}^{(-)}, \tau_{2}^{(-)}...,\tau_{B}^{(-)})$ so that $\tau_{B}^{(-)}=\min\lbrace \tau_{B-1}^{(+)},\tau_{B}\rbrace$ while $\tau_{i}^{(-)}=\tau_{i}^{(+)}$, for  $i=1,..., B-1$. Then, go to \emph{Phase 2} with policy $\pi^{(-)}$.  
\item \emph{Phase 2:} Modify the policy $\pi^{(-)}=(\tau_{1}^{(-)}, \tau_{2}^{(-)}...,\tau_{B}^{(-)})$ from the previous phase to the policy $\pi^{(+)}=(\tau_{1}^{(+)}, \tau_{2}^{(+)}...,\tau_{B}^{(+)})$ so that $\tau_{B}^{(+)}=\tau_{B}^{(-)}$ while $\tau_{i}^{(+)}=\tau_{i}^{(-)}+x$ for  $i=1,..., B-1$ where $x>0$ is the solution of the following:
\begin{equation}
\label{expandtaus}
\phi(\tau_{B}^{(-)}, \tau_{B-1}^{(-)}-\tau_{B}^{(-)}+x,..., \tau_{1}^{(-)}-\tau_{2}^{(-)}+x)=0.
\end{equation}
If $\tau_{B}^{(-)}=\tau_{B}$, the procedure stops and (\ref{expandtaus}) gives the solution that $\phi(\tau_{B}, d_{B-1},..., d_{1})=0$, otherwise go to \emph{Phase 1} with policy  $\pi^{(+)}$.
\end{itemize}
It can be shown that the procedure always stops with a solution that $\phi(\tau_{B}, d_{B-1},..., d_{1})=0$. To see this, first observe that (\ref{expandtaus}) always has a solution as long as:  
\begin{equation}
\label{posphi}
\phi(\tau_{B}^{(-)}, \tau_{B-1}^{(-)}-\tau_{B}^{(-)},..., \tau_{1}^{(-)}-\tau_{2}^{(-)})>0.
\end{equation}
This is due to the following facts about the function $\phi(\tau_{B}^{(-)}, \tau_{B-1}^{(-)}-\tau_{B}^{(-)}+x,..., \tau_{1}^{(-)}-\tau_{2}^{(-)}+x)$: (i) it  is a continuous function of $x$, (ii) it goes to $-\infty$ as $x$ grows.

Next, observe that (\ref{posphi}) always holds, i.e.,
\begin{eqnarray*}
&\phi(\tau_{B}^{(-)}, \tau_{B-1}^{(-)}-\tau_{B}^{(-)},..., \tau_{1}^{(-)}-\tau_{2}^{(-)})=\\
&\underbrace{\phi(\tau_{B}^{(+)}, \tau_{B-1}^{(+)}-\tau_{B}^{(+)},..., \tau_{1}^{(+)}-\tau_{2}^{(+)})}_{=0 \text{ due to the \emph{Step 2} or the initial/optimal policy} }+\int_{\pi^{(+)}}^{\pi^{(-)}}d \phi,
\end{eqnarray*}
is positive. This can be seen by considering:
\begin{eqnarray*}
&\frac{\partial \phi}{\partial \tau_{B}}=2 m_{1}(\tau_{1},\tau_{2},...,\tau_{B})+\\
&\sum_{j=0}^{B-1}\left[  2\tau_{B}\frac{\partial}{\partial \tau_{B}}\mathbb {E}\left[ X \mid E=j\right]-\frac{\partial}{\partial \tau_{B}}\mathbb {E}\left[ X^{2} \mid E=j\right]\right]  \times \\
&\Pr(E=j),
\end{eqnarray*}
which follows from the fact that $\Pr(E=j)$ does not depend on $\tau_{B}$ (see (\ref{transitionerlang})) and can be further simplified by Lemma \ref{diffmoments}, hence:
\begin{equation*}
\frac{\partial \phi}{\partial \tau_{B}}=2 m_{1}(\tau_{1},\tau_{2},...,\tau_{B}).
\end{equation*}
Accordingly, we have:
\begin{eqnarray*}
&\phi(\tau_{B}^{(-)}, \tau_{B-1}^{(-)}-\tau_{B}^{(-)},..., \tau_{1}^{(-)}-\tau_{2}^{(-)})=\int_{\pi^{(+)}}^{\pi^{(-)}}d \phi \\
&=2\int_{\tau_{B}^{(+)}}^{\tau_{B}^{(-)}}m_{1}(\tau_{1}^{(+)},\tau_{2}^{(+)},...,\tau)d \tau > 0,
\end{eqnarray*}
where the inequality follows from the fact that $m_{1}(\tau_{1}^{(+)},\tau_{2}^{(+)},...,\tau)$ being the average inter-update time is always positive.

Therefore, (\ref{expandtaus}) can be always satisfied in \emph{Phase 2}. Also, as the second smallest threshold is strictly increased in \emph{Phase 2}, the smallest threshold can be moved toward $\tau_{B}$ in \emph{Phase 1}. Also, it can be shown that the procedure does not converge any policy other than the policy that $\phi(\tau_{B}, d_{B-1},..., d_{1})=0$. This can be seen considering the following:
\begin{align*}
&\frac{d}{dx}m_{2}(\tau_{1}+x,\tau_{2}+x,...,\hat{\tau}_{B})\mid_{x=0}\\
&<\lim_{x\rightarrow 0}\lim_{n\rightarrow +\infty}\frac{1}{nx}\sum_{k=0}^{n}\left( \mathbb {E}[(X_{k}+x)^{2}]-\mathbb {E}[X_{k}^{2}]\right)\\
&=\lim_{n\rightarrow +\infty}\frac{2}{n}\sum_{k=0}^{n}\mathbb {E}[X_{k}]=2m_{1}(\tau_{1},\tau_{2},...,\hat{\tau}_{B}),
\end{align*}
hence,
\begin{align}
\label{proceduremotion}
&\frac{d}{dx}\phi(\hat{\tau}_{B}, \tau_{B-1}-\hat{\tau}_{B}+x,..., \tau_{1}-\tau_{2}+x)\mid_{x=0}\\
&+\frac{d}{dx}\phi(\hat{\tau}_{B}+x, \tau_{B-1}-\hat{\tau}_{B}-x,..., \tau_{1}-\tau_{2})\mid_{x=0}\\
&>2\hat{\tau}_{B}\frac{d}{dx}m_{1}(\tau_{1}+x,\tau_{2}+x,...,\hat{\tau}_{B})\mid_{x=0},
\end{align}
which implies that the procedure cannot converge to a policy with  $\tau_{B}^{(+)}<\tau_{B}$ as the RHS of (\ref{proceduremotion}) is positive \footnote{This follows from the fact that any increase in thresholds causes an increase in the battery overflow probability which means an increase in the average inter-update duration, i.e, $m_{1}(\tau_{1},\tau_{2},...,\hat{\tau}_{B})$.} and does not vanish for a finite set of thresholds. Therefore, as the smallest threshold of the policies modified by the procedure is increased up to $\tau_{B}$, a solution  that $\phi(\tau_{B}, d_{B-1},..., d_{1})=0$ is eventually reached. This completes the proof. 

\end{document}